%% file: ARXIV_Cleaned.tex
\documentclass[lettersize,journal]{IEEEtran}

\usepackage{amsmath,amssymb,amsfonts,amsthm}

\usepackage{algorithm}
\usepackage{algorithmic}

\usepackage{graphicx}
\usepackage{array}
\usepackage{booktabs}
\usepackage[caption=false,font=normalsize,labelfont=sf,textfont=sf]{subfig}
\usepackage{stfloats}

\usepackage{textcomp}
\usepackage{xcolor}
\usepackage{url}
\usepackage{verbatim}
\usepackage{hyperref}
\usepackage{cite}

\theoremstyle{plain}
\newtheorem{thm}{Theorem}

\theoremstyle{definition}

\newcommand{\regret}{\text{Reg}}
\newcommand{\EE}[1]{{\mathbb{E}}\left[{#1}\right]}

\newcommand{\cov}[1]{{\mathrm{cov}}\left({#1}\right)}

\input{mathcommand.tex}

\begin{document}

\title{A Queueing-Theoretic Framework for Dynamic Attack Surfaces: Data-Integrated Risk Analysis and Adaptive Defense}

\author{
Jihyeon Yun$^{*}$,
Abdullah Yasin Etcibasi$^{*}$,
Ming Shi,
C. Emre Koksal%
\thanks{$^{*}$Jihyeon Yun and Abdullah Yasin Etcibasi contributed equally to this work.}%
\thanks{Jihyeon Yun is with Korea University, South Korea.}%
\thanks{Abdullah Yasin Etcibasi and C. Emre Koksal are with The Ohio State University, USA.}%
\thanks{Ming Shi is with University at Buffalo (SUNY), USA.}
}

\maketitle

\begin{abstract}
We develop a queueing-theoretic framework to model the temporal evolution of cyber-attack surfaces, where the number of active vulnerabilities is represented as the backlog of a queue. Vulnerabilities arrive as they are discovered or created, and leave the system when they are patched or successfully exploited. Building on this model, we study how automation affects attack and defense dynamics by introducing an AI amplification factor that scales arrival, exploit, and patching rates. Our analysis shows that even symmetric automation can increase the rate of successful exploits. We validate the model using vulnerability data collected from an open source software supply chain and show that it closely matches real-world attack surface dynamics. Empirical results reveal heavy-tailed patching times, which we prove induce long-range dependence in vulnerability backlog and help explain persistent cyber risk. Utilizing our queueing abstraction for the attack surface, we develop a systematic approach for cyber risk mitigation. We formulate the dynamic defense problem as a constrained Markov decision process with resource-budget and switching-cost constraints, and develop a reinforcement learning (RL) algorithm that achieves provably near-optimal regret. Numerical experiments validate the approach and demonstrate that our adaptive RL-based defense policies significantly reduce successful exploits and mitigate heavy-tail queue events. Using trace-driven experiments on the ARVO dataset, we show that the proposed RL-based defense policy reduces the average number of active vulnerabilities in a software supply chain by over 90\% compared to existing defense practices, without increasing the overall maintenance budget. Our results allow defenders to quantify cumulative exposure risk under long-range dependent attack dynamics and to design adaptive defense strategies with provable efficiency.
\end{abstract}

\begin{IEEEkeywords}
computer security, vulnerability dynamics, queueing theory, reinforcement learning, long-range dependence
\end{IEEEkeywords}

\section{Introduction}\label{sec:intro}
Cyber risk exhibits temporal dependence and cannot be adequately described by static or stationary reliability models.
Much of the existing approaches in cybersecurity focus on isolated attack models or mitigation mechanisms, offering limited understanding of the holistic and time-varying nature of vulnerabilities that define an organization's \emph{attack surface}. 
Modern infrastructures, spanning cloud services, software-defined networks, Internet of Things and distributed APIs, further amplify these dynamics, producing attack surfaces whose scale and evolution are often unknown even to their operators.

To address this gap, we develop a \emph{dynamic stochastic model} for the evolution of the attack surface. The model generalizes from an individual software component to an entire organization, and ultimately to large-scale ecosystems such as industry sectors or nation-state infrastructures. 
We formalize the instantaneous size of the attack surface as the number of active vulnerabilities, represented by the queue length of a stochastic service process. Arrivals to the queue correspond to the discovery or creation of new vulnerabilities, while departures represent either (a) successful exploitation or (b) successful patching. 
This queueing abstraction makes explicit the role of limited defense capacity, allowing attack-surface management to be studied as a resource-allocation and backlog-control problem.

Building on this foundation, we extend the model to capture the growing influence of automation and AI in both offensive and defensive operations. 
We introduce an \emph{AI amplification factor} that scales vulnerability arrival, exploit, and patching rates. This abstraction is not intended to model AI systems in detail, but to examine how this factor reshapes backlog dynamics.
Our analysis shows that even when attack and defense capabilities scale symmetrically, the rate of successful exploits can still increase superlinearly.

To demonstrate how accurately our proposed framework captures real-world vulnerability dynamics, we apply it to the problem of strengthening open-source software supply chains. Using the ARVO (Atlas of Reproducible Vulnerabilities for Open Source Software) dataset~\cite{mei2024arvo}, which contains over 4{,}000 reproducible vulnerabilities from Google’s OSS-Fuzz platform, we characterize the real-world dynamics of vulnerability discovery and patching across thousands of open-source projects. Event-level analysis reveals that vulnerability arrivals and lifetimes are bursty, heavy-tailed, and non-stationary, and that segmented queueing models accurately reproduce the temporal evolution of the attack surface size across development cycles. This temporal structure further exhibits \emph{long-range dependence (LRD)}, indicating that correlations in exposure decay polynomially rather than exponentially. In practical terms, the effects of individual vulnerabilities persist far beyond their initial disclosure, highlighting systemic bottlenecks in patch deployment and motivating the need for continuous, adaptive, and resource-aware defense strategies to ensure supply-chain resilience.

Motivated by persistent patching delays and the imbalance between vulnerability arrivals and limited defense capacity observed in both data and AI-driven analysis, we develop a reinforcement learning (RL) approach for adaptive defense under resource-budget constraints.
The defense resource, represented by the patching rate, directly influences the service process in our queueing model. Our dynamic framework allows defense rates to vary over time, while explicitly incorporating such switching costs into performance evaluation. Although the resulting control problem is analytically intractable in closed form, we design a low-complexity RL algorithm for adaptive defense allocation under uncertainty. In addition, we rigorously establish a near-optimal regret bound relative to an oracle defender and introduce new switching-reduction techniques that extend the theory of constrained Markov decision processes (CMDPs).

Finally, to illustrate the practical implications of our theoretical and empirical findings, 
we conduct numerical experiments to evaluate the proposed RL-based defense policy. The results show that adaptive resource allocation guided by our RL algorithm can substantially mitigate exploit success rates, achieving reductions of up to 55\% compared to static defense strategies, while maintaining stable performance under both stochastic and adversarial vulnerability arrivals. 
In trace-driven experiments using the ARVO dataset, our adaptive defense policy reduces the average number of active vulnerabilities in a software supply chain by more than 90\% compared to existing defense practices, while operating under the same overall maintenance budget.
These findings underscore that dynamic, learning-based defense policies not only outperform static benchmarks, but also yield smoother and more predictable system behavior. This demonstrates how our analytical framework can directly support real-world cyber-defense decision making.

The main contributions of this paper can be summarized as follows:

\begin{itemize}
    \item \textbf{Dynamic Queueing Model of the Attack Surface:}
    We develop a queueing-theoretic model that captures the temporal evolution of vulnerabilities, where arrivals represent vulnerability discovery and departures represent patching or exploitation.

    \item \textbf{AI-Amplified Threat Dynamics:}
    We introduce an AI amplification factor that scales both attack and defense processes, and show that even symmetric scaling can increase the rate of successful exploits.

    \item \textbf{Empirical Validation and LRD of Vulnerability Dynamics:}
    Using the ARVO dataset, we validate the proposed model and show that heavy-tailed patching times induce long-range dependence (LRD) in the attack surface.

    \item \textbf{Near-Optimal Adaptive Defense via RL:}
    We formulate dynamic defense as a constrained MDP and develop an RL algorithm that achieves near-optimal regret under resource and switching constraints.

    \item \textbf{Defense Switching Cost:}
     We explicitly model the magnitude of policy changes as a switching cost, capturing the operational impact of adjusting defense actions over time. Unlike prior work that penalizes only whether a policy changes, our formulation captures how much the policy changes between consecutive steps. This allows us to quantify the operational cost of large adjustments in defense actions.
\end{itemize}

Together, these contributions establish a quantitative and theoretically grounded foundation for modeling, analyzing, and dynamically defending evolving attack surfaces.

\section{Related Work}

Our work is related to probabilistic approaches to cyber risk analysis.
The industry standard, Factor Analysis of Information Risk (FAIR) framework~\cite{faireframework} formalizes cyber risk quantification through probabilistic factors such as threat events, vulnerabilities, and loss magnitude, providing a common language for risk assessment. Broader treatments of probabilistic cyber insurance and risk evaluation can be found in~\cite{Liu2021CyberInsurance}. While these approaches are influential, they largely assume static system conditions and do not capture the evolving temporal dependencies characteristic of modern attack surfaces.

The concept of the attack surface was formalized by Manadhata and Wing~\cite{surface2010}, and a systematic review~\cite{theisen2018attack} revealed fragmented definitions across hundreds of studies. Recent large-scale analyses, such as~\cite{harry2025measuring}, quantified attack surfaces across government infrastructures, highlighting their scale and complexity. 
These works provide valuable measurement perspectives, but they typically treat the attack surface as a static quantity and do not model how it evolves over time or responds to defense actions.

A related line of work models interdependent vulnerabilities through probabilistic attack graphs~\cite{wang2008attackgraph} and their AI-based extensions~\cite{jin2023prometheus}. Bayesian-network models~\cite{ryan2009bayesian,dynamic2012bayesian,khosravi2020bayesian} have been proposed to estimate compromise probabilities, but these frameworks describe the system at a single snapshot in time. 
We refer to such methods as \emph{snapshot models of risk}, as they capture system state at a fixed point in time and do not represent the sequential or long-range evolution of vulnerabilities.

Efforts to incorporate temporal evolution have used Bayesian networks for industrial and cloud systems~\cite{industrial2018bayesian,graphcloud2022} and Markovian models for sequential attacks~\cite{sequential2019,correlation2022}. 
These studies focus on specific environments rather than the evolution of the attack surface as a whole. Haldar and Mishra~\cite{haldar2017mathematical} and Feutrill et al.~\cite{feutrill2020queueing} observed that vulnerability disclosures exhibit burstiness and long-range dependence, suggesting queueing systems as a natural abstraction. 
However, existing studies do not combine such models with large-scale empirical validation or address the joint temporal and spatial dynamics of vulnerability backlogs.

A key enabler for such modeling is the availability of event-level vulnerability data. The recently released ARVO dataset~\cite{mei2024arvo} provides detailed timestamps of vulnerability discovery and patching. Our work is the first to leverage ARVO to calibrate and validate a queueing-theoretic model of attack surface evolution, bridging theoretical abstractions with empirical vulnerability dynamics.

The rapid integration of AI into both software development and exploitation further complicates this landscape. While large language models (LLMs) can assist in code repair~\cite{berabi2024deepcode,staab2024large}, they also accelerate exploit generation~\cite{fang2024llm,xie2024gradsafe,geiping2024coercing}. Reports by practitioners and agencies~\cite{miessler_linkedin_2024,fbi_ai_cyber_threat_2024} highlight this dual role of AI as both attacker and defender. 
Yet existing models do not provide a quantitative framework for studying how automation simultaneously affects vulnerability discovery, exploitation, and patching dynamics. Our use of an \emph{AI amplification factor} is intended to capture these rate-level effects in a tractable way.

Finally, constrained and safe RL has been studied under budget \cite{amani2021safe,miryoosefi2022simple,shi2023nearconstraint} and policy-adaptation \cite{bai2019provably,huang2022towards,shi2023nearswitch} constraints. Existing studies on policy-adaptation primarily penalize the number of policy changes, i.e., the frequency of updates. Without the magnitude of change, it is not completely possibly to quantify the operational cost of change actions in practical defense settings. In contrast, our formulation models the amount of change in the executed defense action and quantifies the magnitude of consecutive policy adjustments. 
This distinction allows us to model reconfiguration overhead in a more realistic way.
Moreover, previous work does not consider the joint effect of resource-budget and switching-cost constraints on adaptive defense policies.
Our formulation unifies these elements and provides the first theoretical regret guarantees for RL in this setting.

Overall, our study introduces a queueing-theoretic perspective that explicitly models the time-varying and heavy-tailed nature of vulnerability backlogs. By validating the model on real data and integrating it with adaptive control, we provide a quantitative framework for analyzing dynamic attack surfaces and defense resource allocation. Using the framework, we provide an RL-based systematic approach to allocating constrained defensive resources to achieve a significantly improved attack surface dynamics, with the variations in the budget directly taken into account.

\section{System Model}
\label{sec:model}

Consider a single component in an organization’s IT stack, such as an authentication service, file server, or endpoint device. Each component maintains an attack surface, which represents the \textbf{set of currently active vulnerabilities}. For example, in a software release, the attack surface can be defined as the set of unpatched bugs (the definition can also be extended depending on the dependencies to the other systems). In a larger ecosystem, like an enterprise, the IT/OT environment will have a complex attack surface, composed of the combination of the attack surfaces of each component in the system. 
The overall attack surface will exhibit an interplay across the network of components, and its size naturally reflects how many vulnerabilities remain exposed at a given time.
In this paper, to build the initial foundation, we focus on a single subsystem or component. The single-component model can be naturally extended to multi-component or multi-organizational settings; we briefly discuss such extensions in Section~\ref{sec:discussion}.

The size of the associated attack surface at time $t$ is represented by the stochastic process $N(t)$, which we model as the number of jobs in a queue. Here, arrivals correspond to the appearance of new vulnerabilities, and services represent their removal through patching or exploitation\footnote{{It is straightforward to work out all our results (with insights unchanged) to the alternative setting where a vulnerability remains in the system after exploitation and may be exploited multiple times before being patched.}}.
In our formulation, a vulnerability is treated as a job that remains in the system until it is either patched or successfully exploited. While in practice a vulnerability may persist and be exploited multiple times, our abstraction considers the first successful exploit as the completion of the corresponding task.
Let $V(t)$ denote the arrival process of vulnerabilities, and $N_d(t)$ and $N_l(t)$ denote the instantaneous numbers of defended and exploited vulnerabilities at time $t$, respectively, as shown in Fig.~\ref{fig:model}.
The discrete-time evolution of the attack surface size is given by
\begin{equation}
    N(t+1) = \bigl\{ N(t) + V(t) - [N_d(t) + N_l(t)] \bigr\}^+,
    \label{eq:evolution}
\end{equation}
where $\{\cdot\}^+ = \max\{\cdot, 0\}$.
This ensures that the queue length remains nonnegative.
Modeling the attack surface as a queue highlights backlog dynamics, with vulnerabilities accumulating when arrival rates exceed patching capacity and decreasing only when defenses are sufficient. This recursion embodies the key intuition: new vulnerabilities enlarge the attack surface, while patching and exploitation act as concurrent removal mechanisms. As shown later in our empirical analysis (Section~\ref{sec:arvo}), both the arrival and lifetime processes exhibit burstiness, heavy-tailed persistence, and non-stationarity.

\begin{figure}[t]
\centering
\includegraphics[width=\columnwidth]{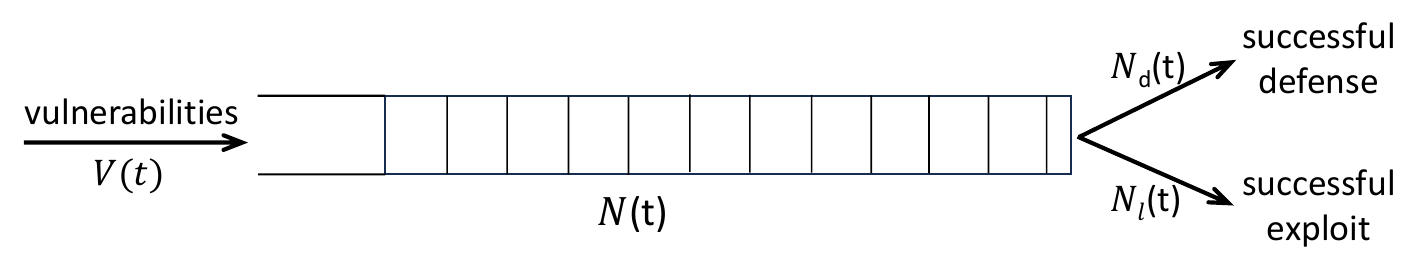}
\caption{Attack surface modeled as a queueing system.
Vulnerabilities arrive via $V(t)$ and depart through competing defense (patching) and exploit processes.}
\label{fig:model}
\vspace{-0.15in}
\end{figure}

Each vulnerability is subjected to a \emph{race condition} between defensive and offensive actions. Let $D_d$ and $D_l$ denote random variables representing the defense time and exploit time, respectively. For each active vulnerability,
\begin{equation}
    D_s = \min\{D_d, D_l\},
\end{equation}
determines its completion time, and the winner of the race
$s = \argmin\{D_d, D_l\}$ increments the corresponding counter $N_s(t)$.
For instance, if $D_d = 10^{-2}$ and $D_l = 10^{-3}$ for a given vulnerability,
the attacker acts ten times faster, leading to an exploit departure.
This race captures the operational reality that a vulnerability remains exposed until either it is patched or it is exploited.

Lastly, we denote the instantaneous 
\emph{allocated} service rates of the defense and exploitation processes 
at time $t$ by $\mu_d(t)$ and $\mu_l(t)$, respectively. These quantities represent the total resources devoted to patching and exploitation.

Although $N_d(t)$ and $N_l(t)$ denote the realized numbers of vulnerabilities that are patched and exploited at time $t$, respectively, their expectations $\mathbb{E}[N_d(t)]$ and $\mathbb{E}[N_l(t)]$ do not, in general, coincide with $\mu_d(t)$ and $\mu_l(t)$. This is due to the race condition between the two processes.

In particular, when both attack and defense are active, the realized rates are jointly determined by the competition between the two processes, and we have
\[
\mathbb{E}[N_d(t)] \leq \mu_d(t), \qquad 
\mathbb{E}[N_l(t)] \leq \mu_l(t),
\]
with equality holding only in the absence of competition. Specifically, 
$\mathbb{E}[N_d(t)] = \mu_d(t)$ if and only if $\mu_l(t)=0$, and 
$\mathbb{E}[N_l(t)] = \mu_l(t)$ if and only if $\mu_d(t)=0$.

To model defensive resource limitations, we consider a system with $m \in \mathbb{N}$ homogeneous parallel servers, where each server processes at most one vulnerability at a time. Thus, at any time $t$, up to $m$ vulnerabilities can be processed concurrently, i.e., $N_d(t) \leq m$. Vulnerabilities are assigned to servers according to a work-conserving policy, so that any idle server immediately begins processing a waiting job if available.

In addition to this concurrency constraint, we impose a global service-rate constraint $\mu_d(t) \leq b$ on the aggregate defense effort. To put it together, the defense process satisfies
\begin{equation}
   \mu_d(t) \leq b, \qquad N_d(t) \leq m,
   \label{eq:resource_constraint}
\end{equation}
where $b$ represents the maximum total service rate that can be sustained across all active servers. This captures shared resource limitations that bound the total achievable patching throughput, even when multiple servers are active. In practice, defensive prioritization or limited concurrency can be modeled by reducing $m$ or adjusting per-server service rates.

These characteristics motivate the adoption of a general $G/G/m$–$b$ model rather than simpler memoryless abstractions. Note that in Kendall's notation, the standard $G/G/m/k$ framework \cite{gautam2012analysis} uses $k$ to denote the maximum number of jobs allowed in the system (i.e., a queue-length capacity constraint). In contrast, our $G/G/m-b$ notation utilizes $m$ to denote the number of parallel servers and $b$ to represent the aggregate capacity constraint imposed on the total service rate, effectively modeling resource-limited defense operations.
This formulation directly captures scenarios where multiple vulnerabilities can be queued but the total patching effort is fundamentally limited by shared resources.

Unlike classical queueing models with independent service rates, both $\mu_d(t)$ and $\mu_l(t)$ may depend on the current attack surface size $N(t)$. As $N(t)$ grows, defenders must divide limited resources across more vulnerabilities, while attackers benefit from the expanded surface.
This coupling creates a feedback effect: when the number of active vulnerabilities increases, the same defense capacity must be shared across more items, slowing down patching on each vulnerability, while attackers face more exposed targets and thus have more opportunities to succeed.
For instance, $\mu_d(t)$ may decrease inversely with $N(t)$, while $\mu_l(t)$ increases proportionally to $N(t)$,
reflecting the asymmetric scalability of attack versus defense.
The interplay between these processes governs the temporal evolution of the attack surface.

\noindent \textbf{Variations of the Model:} Throughout this paper, we consider several specializations derived from our model:

\begin{itemize}
    \item \textbf{Temporal variation analysis:} 
    We use the limiting case $M/G/\infty$, which isolates temporal effects such as heavy-tailed persistence and LRD without capacity constraints, to build a theorem on how the heavy-tailed nature of the arrival and service processes affect the attack surface variations.
    \item \textbf{Data integration:} 
    In Section~\ref{sec:arvo}, the model is instantiated as $G/G/m$–$b$ in its full generality to capture bounded defense capacity and bursty vulnerability arrivals observed in the ARVO dataset.
    \item \textbf{Dynamic defense design and optimization:} 
    In Section~\ref{sec:RL}, we consider a setting where the defense rate $\mu_d(t)$ is adjusted over time, subject to the constraint $\mu_d(t) \le b$. 
    Here, $b$ is fixed, and the defense behavior changes through the time-varying control of $\mu_d(t)$.
    
\end{itemize}

\section{Problem Formulation}
\label{sec:problem}

Building on the stochastic queueing model above, we now formulate the adaptive defense problem. The objective is to allocate limited defense resources over time to minimize long-term exposure and breach costs, while accounting for reconfiguration (switching) overhead.

At each time step $t$, the defender selects a defense (patching) rate $\mu_d(t)$ subject to the resource-budget constraint $\mu_d(t)\!\le\!b$, while the effective exploitation rate $\mu_l(t)$ evolves according to the coupled arrival–service dynamics defined earlier. 
The resulting queue length $N(t)$ captures the number of active vulnerabilities and thus represents the instantaneous \emph{attack surface size}. 
The control task is to design a policy $\pi=\{\mu_d(t)\}_{t=1}^T$ that balances:
(i) risk reduction through faster patching, 
(ii) efficiency in total resource use, and 
(iii) stability against frequent reallocations.

We study two core problems that together form the foundation of our framework. The first focuses on data-driven model inference and empirical validation, while the second builds on the first and develops an adaptive control policy for dynamic defense allocation. Each problem highlights a distinct analytical or algorithmic component of the overall approach:

\begin{enumerate}

    \item[\textbf{(P1)}] \textbf{Data-Driven Characterization.}

    Given event-level vulnerability data containing discovery and patch timestamps, we model the attack surface as a non-stationary stochastic process. Rather than fitting a single stationary distribution, we seek to identify a segmentation of the time horizon into quasi-stationary intervals, together with segment-wise model parameters.
    Given event-level vulnerability data containing discovery and patch timestamps, we model the attack surface as a non-stationary stochastic process. Rather than fitting a single stationary distribution, we seek to identify a segmentation of the time horizon into quasi-stationary intervals, together with segment-wise model parameters. Let $\{\mathcal{T}_k\}_{k=1}^K$ denote a partition of the time axis into $K$ segments, and let $\hat{P}_k$ denote the empirical queue-length distribution (QLD) in segment $\mathcal{T}_k$. We jointly estimate the number of segments $K$, the segmentation $\{\mathcal{T}_k\}_{k=1}^K$, and the model parameters $\{\theta_k\}_{k=1}^K$ by solving
    \begin{equation}
        \min_{K,\,\{\theta_k\},\,\{\mathcal{T}_k\}} \;\sum_{k=1}^K d\big(\hat{P}_k, P(\theta_k)\big),
    \end{equation}
    where $d(\cdot,\cdot)$ is a divergence metric (KL divergence in our implementation), $\theta_k \in \Theta$ parameterizes the $G/G/m$-$b$ queueing model in segment $k$, with $\theta_k = \{m, b, F_{IA}, F_{ST}\}$.
    Solving (P1) yields a segmented and validated model, where each quasi-stationary segment is parameterized independently, capturing the non-stationary and heavy-tailed behavior of vulnerability dynamics observed in real data. The number of segments $K$ is selected in a data-driven manner based on diminishing returns in the fitting error (e.g., via an elbow criterion), ensuring a balance between model accuracy and complexity. This segmented representation provides the empirical foundation for the adaptive defense control in (P2).

    \item[\textbf{(P2)}] \textbf{Learning-Based Adaptive Defense.} 
    
    Using the quasi-stationary segments identified in (P1), we now formulate the adaptive defense problem. We aim to develop a learning policy that adaptively controls $\mu_d(t)$ to optimize defense resource allocation.
    The objective is to minimize the cumulative cost, given by
    \begin{align}
        &\min_{\{\mu_d(1:T)\}}
        \sum_{t=1}^{T} 
        \mathbb{E}\Big[
            C(N(t))
            + |\mu_d(t)|
            \nonumber\\[-4pt]
            &\qquad\qquad
            + g\!\left(|\mu_d(t)-\mu_d(t-1)|\right)
        \Big] \label{eq:rl_obj}\\
        &\text{sub. to:}\quad 
        0 \le \mu_d(t)\le b,\quad t=1,\dots,T . \nonumber
    \end{align}
    where $|\cdot|$ denotes the absolute value. The expectation reflects stochastic variability in attack arrivals and patching delays. The first term $C(N(t))$ penalizes successful exploits proportional to the instantaneous attack surface size, the second term $\mu_d(t)$ captures the defense effort at time $t$, and the third term $g(\cdot)$ models the switching cost for defense reconfiguration. Unlike the abstract policy-adaptation cost used in existing works, this switching cost is proportional to the \emph{magnitude} of change in the executed defense action. The second term directly captures the cumulative defense effort over time, corresponding to total patching resources. In contrast, the argument of $g(\cdot)$ measures the change in defense rate between consecutive time steps, i.e., $|\mu_d(t)-\mu_d(t-1)|$, reflecting the operational cost of reconfiguring defense actions rather than the frequency of policy updates. The learned model from (P1) provides the system dynamics and resource parameters used in (P2), thereby enabling a data-driven approach to optimizing defensive resource allocation.
\end{enumerate}

\textit{Remark:} In this work, we assume that all vulnerabilities are treated with equal priority. Our model does not explicitly differentiate vulnerabilities based on severity or potential impact. Instead, the system-level effect of vulnerabilities is captured through the cost function, which reflects the aggregate impact of successful exploits. Incorporating heterogeneous vulnerability classes is also an important extension and we leave it for future works.
    
Before introducing adaptive defense strategies, we first analyze a few simple scenarios under basic static allocation case to build some intuition on attack surface dynamics.

\section{Illustrative Examples: Static Resource Allocation}
\label{sec:static}

We begin with a simple baseline that assumes a fixed defense allocation. This choice is \textbf{not} intended to represent a realistic situation, but rather to build intuition about how defense capacity and vulnerability arrivals interact in a queueing system. In this example, we assume memoryless arrivals and departures from our queue. These insights will help us better interpret the results in the later sections, where we relax the memoryless assumption.

\subsection{M/M/$\infty$ Abstraction}
\label{sec:m-m-inf}

The memoryless nature of the arrival and service processes leads to an M/M/$\infty$ queue, which removes temporal correlations and capacity interactions, allowing us to focus on how vulnerability arrivals and fixed defense capacity jointly determine attack surface size and exploitation rates.

In our evaluations, we use a sampled discrete-time version of the continuous-time M/M/$\infty$ system, consistent with the discrete-time evolution in Eq. (\ref{eq:evolution}). The resulting system can be represented as a Markov chain with a countable state space.

As described in Section~\ref{sec:model}, in our analyses we assume that organizations have fixed amount of cyber resources and allocate the full amount without a variation from one episode to another. In particular, the rate of the successful defense process $\mu_d = \alpha \lambda$ remains constant and thus independent of $N(t)$ where $\lambda$ is the arrival rate for $V(t)$, the (stationary) vulnerability process, and $\alpha$ is a constant such that $\alpha \lambda \leq b$. Here, $\mu_d$ denotes the total defense rate.
The variable $\alpha$ signifies the intensity of the defense. As an example, if $\alpha=1$, we say the defense rate is $100$\%  or if $\alpha=0.5$, we say the defense rate is $50$\%.

\begin{figure}[!t]
  \centering
  \includegraphics[width=0.7\linewidth]{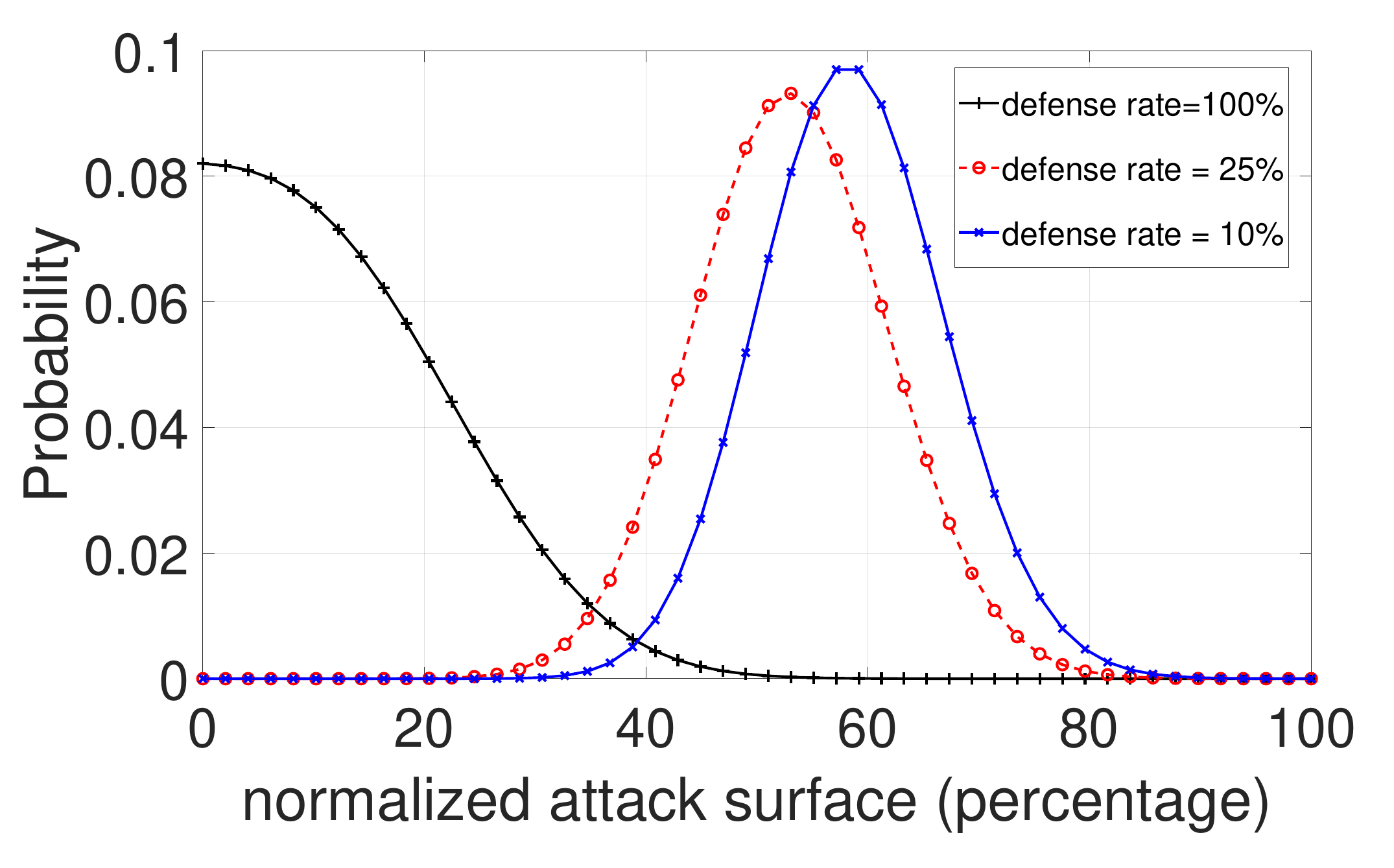}
  \caption{Steady-state distribution of the normalized attack surface size under different defense rates. Lower defense rates shift the distribution toward larger attack surfaces, increasing exposure risk. The surface size is normalized to the range $[0,100]\%$ for visualization.}
  \label{fig:attack_surface_distribution}
\end{figure}

\begin{figure}[!t]
  \centering
  \includegraphics[width=0.6\linewidth]{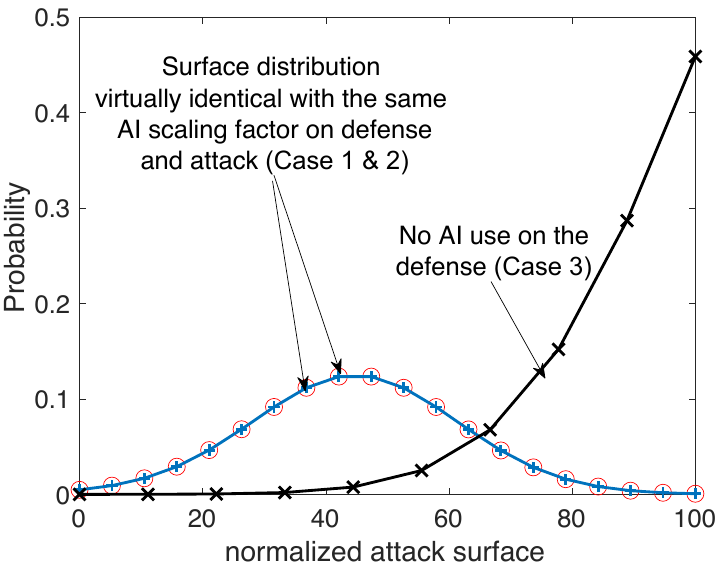}
  \caption{Steady-state distribution of the normalized attack surface size under different AI usage scenarios. Symmetric AI scaling leaves the distribution unchanged, whereas AI applied only on the attack side (no AI on defense) shifts the distribution toward larger attack surfaces, increasing exposure risk.}
    \label{fig:attack_surface_distribution_AI}
\end{figure}

On the attacker side, we assume the rate of the successful exploitation process grows proportional to $N(t)$ as a larger number of active vulnerabilities attracts more attack attempts targeting the exposed surface, growing proportional to the attack surface size. Hence,
\( \mu_l(t) = \beta \lambda N(t), \)
where $\beta$ is a constant denoting the intensity of attacks on the organization. 

Note that this is not a strict M/M/$\infty$ model because the service rate includes both load-independent and load-dependent components (i.e., $\lambda$-independent and $\lambda$-dependent). We use this model as a simple abstraction to illustrate the key dynamics of the system.

In Fig.~\ref{fig:attack_surface_distribution}, we illustrate the steady-state distribution of the attack surface size. Since the underlying system is unbounded, the queue length is normalized by a reference value (e.g., a representative steady-state level) and expressed as a percentage for visualization purposes. This normalization provides a relative measure of the attack surface size and does not impose any artificial bound on the system. We set the vulnerability arrival rate to $\lambda = 100$ per unit time and the attack rate to $\beta = 0.001$. The figure shows the resulting distributions for three defense rates: $\alpha = 100\%$, $25\%$, and $10\%$.

The curve with the full ($100$\%) defense rate leads to a small expected surface size of $\EE{N(t)}=13.2$\% and a time-averaged breach rate of $\lim_{t \to \infty}\frac{1}{t}\EE{N_l(t)}=6.79$ breaches per unit time, much lower than the time-averaged defense rate, which happens to be
\[ \lim_{t \to \infty}\frac{1}{t}\EE{N_d(t)}=\lambda-\lim_{t \to \infty}\frac{1}{t}\EE{N_l(t)}=93.21. \] 
As the defense rate decreases, the expected surface size and exploitation rate increase sharply.
At $\alpha=50$\%, the expected surface size grows to $52.1$\% with a substantial breach rate of $69.69$. This means, there are more than twice as many breaches as there are successful defenses.  Breach rate grows even further to $84.16$\% as we further decrease the defense rate to $10$\%. The interesting thing here is that, the expected attack surface size in this final situation remains at $57.44$, very close to the case with $\alpha=0.25$, despite a significant decrease to $\alpha=0.1$.

The above situation can be explained by the fact that, as the surface size grows, much of the departures are caused by a breach rather than a successful patch. As a result, the average surface size remains relatively static, however the instantaneous fluctuations above the mean are quickly exploited by the growing set of attackers. As a result, even though one may think that the attack surface is not much higher, the breach rate grows at a much higher pace with reduced defense rate.

Our initial results demonstrate a \textbf{phase-transition} phenomenon in the attack surface size distribution. Once the defense rate goes below a certain point, the surface distribution shifts sharply and abruptly to the right. Further reducing the defense rate beyond that shifting point does \textit{not} change the distribution considerably. 
This observation underlines the importance of keeping a disciplined security posture for an organization and the resources should be allocated to have a rate at least identical to (if not much higher than) the rate of growth in vulnerabilities.  
For example, if the defense rate is set to $\alpha = 200$\%, the expected attack surface size remains below \textit{a single vulnerability}, while the average breach rate is merely $0.29$ per unit time! Our results show that a small increase in resources can substantially improve protection against breaches, whereas insufficient resource allocation leads to sharply degraded security performance.

\subsection{AI-driven Dynamics}
\label{sec:ai_driven}

We continue the illustrative analysis with the memoryless model to study how AI-driven acceleration of vulnerability discovery and response affects attack-surface dynamics, captured through a simple rate-scaling abstraction.
Let us introduce an \emph{AI amplification factor} that scales the arrival and service rates and study its effect under symmetric and asymmetric amplification of attack and defense capabilities. We scale the vulnerability arrival rate and the exploitation rate by the same factor $a$, so $\lambda \rightarrow a\lambda$ and $\mu_l\rightarrow a \mu_l$. In the second part, we will scale the defense rate $\mu_d$ with the same rate, to evaluate the impact of the use of AI on the defensive side, as well as the attackers' side.

In this example, we use the same amplification factor across the three pillars of the model. Our intention here is to illustrate the drastic shift in temporal dynamics even when there is no change in the spatial dynamics of the surface. Also, we show the dynamics under asymmetry in AI amplification between attack and defense, where the asymmetry is in the favor of the attackers.

Similar to Section~\ref{sec:m-m-inf}, we use $\mu_d(t)=\alpha \lambda$, where $\alpha$ is the defense rate and $\mu_l(t)=\beta \lambda N(t)$. We provide the probability mass function for the attack surface for three different situations:
\begin{enumerate}
    \item No AI is used (i.e., $a=1$) on either the offense or the defense. Here, we choose the vulnerability arrival rate $\lambda=5$, defense rate $\alpha = 50\%$, and the attack rate $\beta=0.005$;
    \item AI used on the attack and defense with an AI amplification factor of $a=4$ on both sides. Here, the $\lambda = 20$, is amplified by $a$ compared to the previous case and both the attack and defense resources are also benefiting the same rate of amplification.
    \item AI is used on the attack side only. As a result, the overall defense resources remain at $2.5$ units as in the original situation, while the vulnerability arrival rate and the attack rate scale with the AI amplification factor $a=4$. 
\end{enumerate}
In Fig.~\ref{fig:attack_surface_distribution_AI}, we illustrate the attack surface distributions for Cases (1-3) above. For Case 1 (no AI), the expected surface size remains at $\EE{N(t)}=42.48$\% at a defense rate of $2.5$ patches per unit time while the exploit rate is $2.06$ exploits per unit time.

Notably, when AI-driven acceleration is applied symmetrically to both attack and defense, the steady-state distribution of the attack surface remains unchanged, even though vulnerabilities arrive and are processed at a faster rate. However, all event rates scale with the amplification factor. In particular, \textit{the exploitation rate increases from $2.06$ to $8.24$ exploits per unit time}, i.e., by a factor of $a$. Thus, while the shape of the attack-surface distribution is preserved, successful exploits occur more frequently due to the accelerated underlying dynamics. 
This observation highlights that symmetric acceleration primarily compresses the time scale of events rather than altering the distribution itself, and suggests that simply matching attack acceleration with defensive acceleration may be insufficient to reduce exploitation frequency.

Lastly, if the defense does not use AI, while the AI is used on the attack, the expected surface size substantially increases to $80.54\%$. The exploitation rate is scaled up to $13.23$ per unit time, demonstrating a super-linear increase with AI amplification factor. This observation shows that an asymmetry in the AI usage in the favor of the attack side leads to a disproportionately higher increase in the rate of successful exploits. 
This scenario illustrates how asymmetric acceleration on the attack side can significantly worsen backlog and exploit rates under fixed defense capacity.

The illustrative examples above demonstrate how fixed and AI-amplified defense rates influence the steady-state behavior of the attack surface. We now turn to real-world data to assess whether these modeled dynamics hold in practice. In the next section, we integrate empirical vulnerability data from the open-source software repositories and validate our queueing-theoretic framework against observed attack surface behavior.

\section{Data Integration: Software Supply Chain}
\label{sec:arvo}

In this section, we apply our queueing-theoretic framework to empirical vulnerability data in order to map the dynamics of an actual attack surface onto the model. Our objective is to identify how observable quantities including vulnerability disclosures, patching times, and backlog evolution project onto the underlying queueing variables. For that purpose, we estimate the corresponding system parameters and resource constraints using available data. Once such a correspondence is established, the model serves as a basis for reasoning about defense strategies and mitigation policies.

To that end, we implement our framework in the use case of open-source software supply chain. In our implementation, we use the ARVO dataset~\cite{mei2024arvo}, which aggregates more than 4{,}000 reproducible vulnerabilities from Google’s OSS-Fuzz infrastructure, spanning hundreds of large-scale open-source C/C++ projects. Each record includes rich metadata such as report and fix timestamps, sanitizer type (ASan, MSan, UBSan), crash category (for example, heap buffer overflow or use after free), and severity level (low, medium, or high). This information enables precise event-level tracking of vulnerability discovery and patching. The dataset’s granularity makes it particularly well suited for queueing-based modeling: vulnerability disclosures correspond to \emph{arrivals}, while patch completions represent \emph{service completions}.

Using this dataset, we first analyze the empirical dynamics of vulnerability arrival and departures, demonstrating that our queueing-theoretic framework provides an accurate and interpretable representation of real-world attack surface evolution. We show, using KL divergence and comparative fits, that the queueing model accurately reproduces the observed queue size dynamics with near-empirical precision, thereby validating it as a realistic abstraction of complex software ecosystems. We further characterize the heavy-tailed nature of both arrival and service processes, which reveals a systemic bottleneck that slows patch deployment and motivate the need for adaptive, data-driven defense strategies. In the next section, we build on these findings and propose a RL–based dynamic defense allocation algorithm that optimally distributes defensive effort to manage the attack surface size under resource and switching constraints.

\subsection{Validating the Proposed Queueing Model on the ARVO Dataset}

The following steps outline our complete empirical pipeline for constructing, segmenting, and validating the queueing model on the ARVO dataset, thereby linking theoretical formulation with real-world vulnerability dynamics.

\textbf{Step 1. Queue reconstruction and exploratory analysis:}
We first align vulnerability discovery and patching timestamps to reconstruct the time series of open vulnerabilities $N(t)$.
The ARVO dataset used here provides exceptionally high resolution, tracking over $4410$ vulnerabilities across $260$ unique open-source projects from December 2016 to May 2024. Each record includes exact event-level timestamps for vulnerability discovery and patching, alongside critical metadata such as severity (including $1150$ High-severity cases), detection sanitizer (e.g., asan, msan), and specific crash types like Heap-buffer-overflow and Use-of-uninitialized-value. 
Figure~\ref{fig:queue_vs_time} shows clear expansion and contraction phases, with bursty arrivals followed by delayed patching, confirming non-stationarity due to capacity limits in patching throughput.

\begin{figure}[t]
    \centering
    \includegraphics[width=0.8\linewidth]{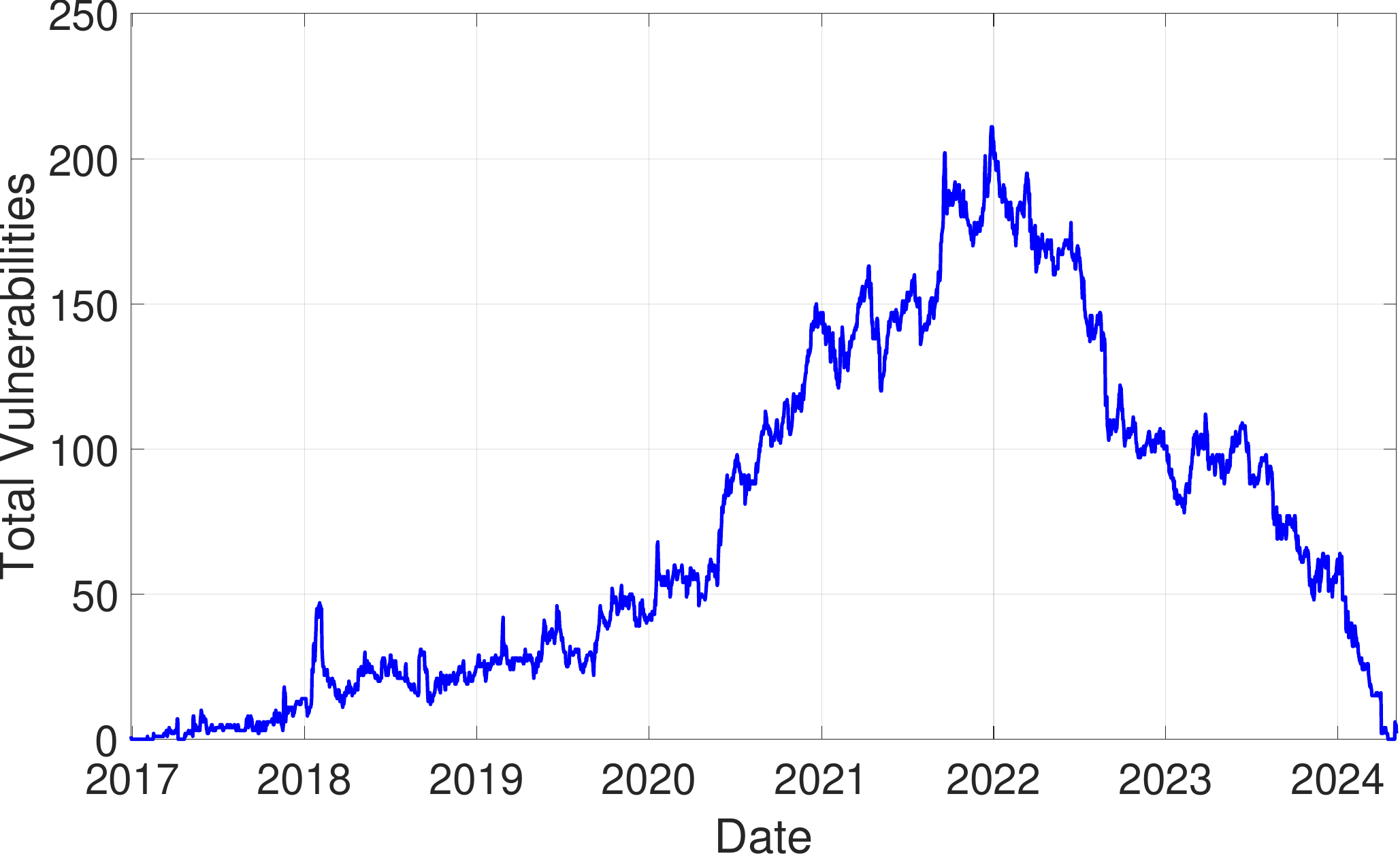}
    \caption{Temporal evolution of the attack surface size, $N(t)$, in the ARVO dataset, showing bursty discovery, delayed patching, and non-stationary behavior.}
    \label{fig:queue_vs_time}
\end{figure}

\begin{figure}[t]
    \centering
    \includegraphics[width=0.8\linewidth]{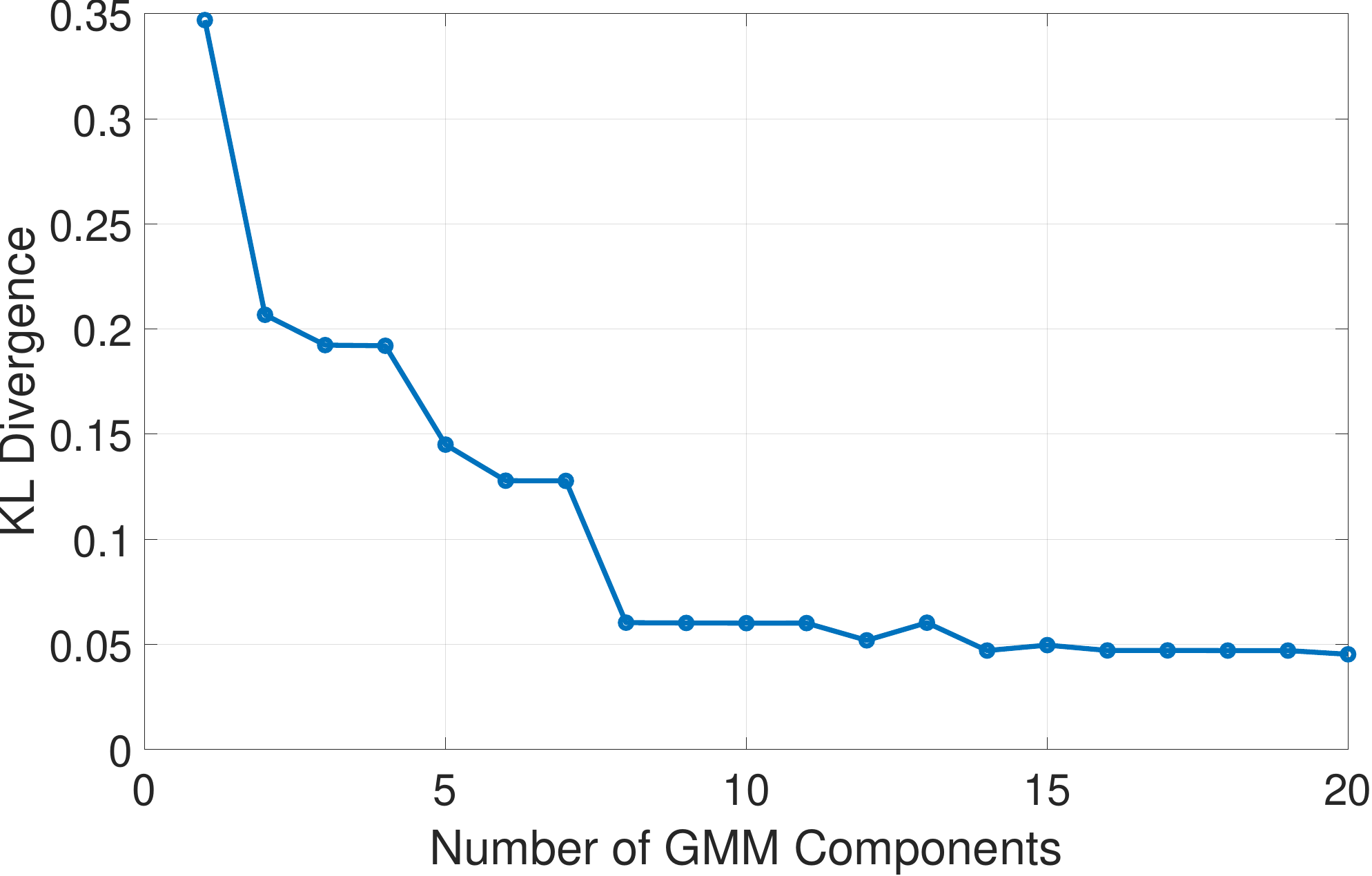}
    \caption{KL divergence versus the number of Gaussian mixture components. Fit quality improves rapidly up to about ten components, after which additional components yield diminishing returns.}
    \label{fig:kl_vs_components}
    \vspace{-1em}
\end{figure}

\textbf{Step 2. Segmentation via Gaussian mixture modeling:}
To capture these evolving patterns, we employ a \emph{segmented modeling approach}. 
Segments are defined as mutually exclusive, closed intervals $[t_{start}, t_{end}]$ that collectively partition the entire observation period. Within each interval, arrivals and departures are analyzed independently to capture localized distributional shifts and non-stationarities in the time series. 

To identify these segments in a data-driven manner, we use a Gaussian mixture model (GMM) fitted to the empirical QLD. The GMM captures multimodal structure in the data, where each component corresponds to a distinct quasi-stationary regime. This allows us to decompose the non-stationary process into a collection of locally quasi-stationary segments.

The number of mixture components is selected based on the KL divergence elbow curve shown in Figure~\ref{fig:kl_vs_components}, which indicates that model fit improves significantly up to around ten components and then saturates. This segmentation enables accurate modeling of non-stationary dynamics through locally stationary approximations.

\textbf{Step 3. Segment-wise parameter estimation:}
Within each segment, we estimated inter-arrival and service distributions and calibrated the queue parameters \((m,b)\) of a \(G/G/m\)–\(b\) model by minimizing the KL divergence between empirical and simulated QLDs. In this segmented setting, the parameter \(b\) represents the \emph{mean available defensive resource} rather than the maximum capacity used in the next section, reflecting the average effective throughput observed in each operational regime. The resulting segmented models accurately reproduced the multimodal and time-varying dynamics of the attack surface, confirming that segmentation is essential for representing the non-stationary evolution observed in ARVO.

\textbf{Step 4. Statistical characterization of IA and ST:}
After segmentation, we analyzed the stochastic structure within each stationary window 
to identify appropriate parametric distributions for inter-arrival ($F_{IA}$) and service times ($F_{ST}$). We evaluated a wide range of candidate distributions using five divergence metrics (KL, TVD, L2, JSD, and Wasserstein). 
Heavy-tailed mixtures consistently outperformed non-heavy-tailed models, such as the exponential distribution, which underestimated tail mass and failed to capture persistence effects. As illustrated in Figure \ref{fig:comp1_ia_st_fits}, for the first segment (weeks 0–64), the best-fitting non-heavy-tailed model (exponential) yielded significantly higher KL divergences of 1.31 for IA and 1.34 for ST. In contrast, the heavy-tailed loglogistic and Gamma–InverseGaussian distributions achieved much lower KL divergences of approximately 0.77 and 0.42, respectively.

Intuitively, a heavy-tailed distribution implies that very large delays are not rare. This is evident in Fig. \ref{fig:comp1_ia_st_fits}, where the empirical curves remain substantially above the exponential fit at large times, indicating significantly more probability mass in the tail.

These results confirm that both vulnerability discovery and patching processes are fundamentally heavy-tailed, justifying our use of more general $G/G/m$ abstractions to capture long-lived exposure and temporal clustering.

\begin{figure}[t]
    \centering
    \includegraphics[width=\columnwidth]{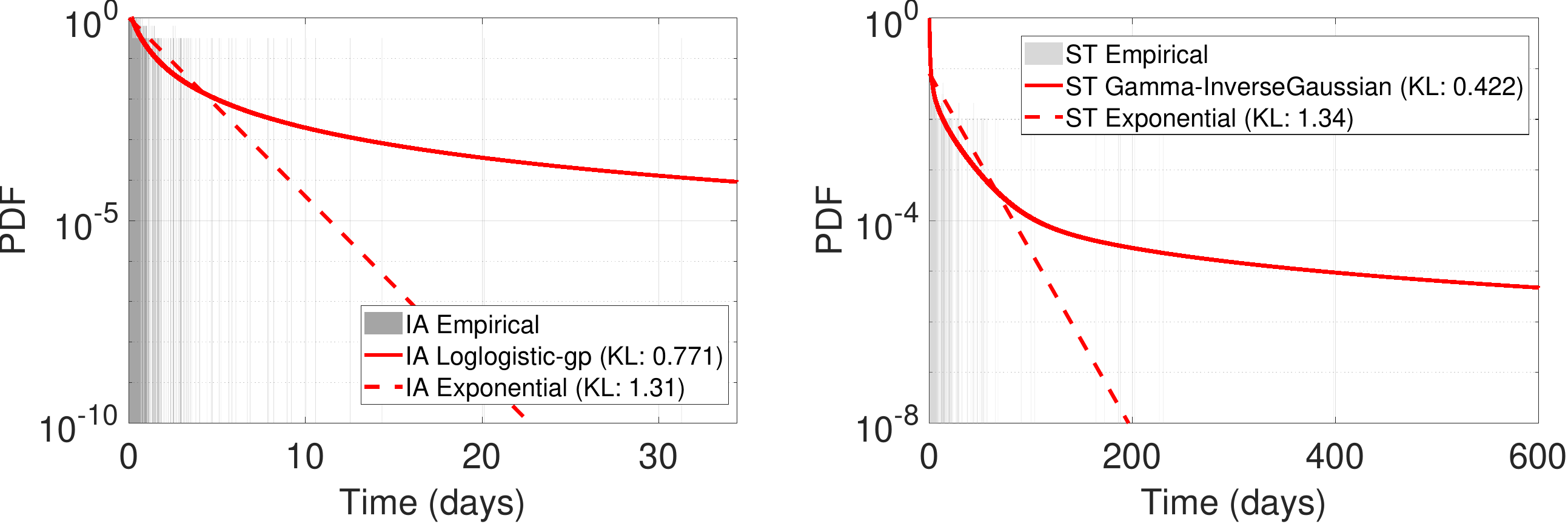}
    \caption{IA and ST distributions for Component 1 (weeks 0–64). 
    Loglogistic-general Pareto fits IA (KL $\approx$ 0.77); Gamma–InverseGaussian fits ST (KL $\approx$ 0.42).}
    \label{fig:comp1_ia_st_fits}
    \vspace{-0.1in}
\end{figure}

\textbf{Step 5. Segment-wise queue model fitting and validation:}
Finally, we validate that the segmented $G/G/m$–$b$ model accurately reproduces the empirical QLD observed in ARVO. Figure~\ref{fig:final_qld} compares the empirical QLD with the segmented, bootstrap, and ten-component GMM fits. In the bootstrap model, samples are drawn directly from the empirical data rather than from any fitted distribution, providing a nonparametric characterization. As shown, the segmented $G/G/m$–$b$ model reproduces the empirical QLD with high fidelity, accurately capturing both the multimodal structure and the heavy-tailed persistence observed in practice. Quantitatively, the KL divergence between the empirical and simulated distributions is $\mathbf{0.1072}$, comparable to the nonparametric bootstrap ($\mathbf{0.1058}$).

Across segments, the number of servers $m$ remains relatively stable ($220$–$250$), while the effective resource capacity $b$ varies widely ($30$–$270$), reflecting changes in patching throughput and organizational defense posture. Together, these results confirm that \textit{empirically calibrated queueing abstractions replicate real-world attack surface dynamics with sub–$0.11$ KL divergence, demonstrating near-empirical precision}.

\begin{figure}[t]
    \centering
    \includegraphics[width=\columnwidth]{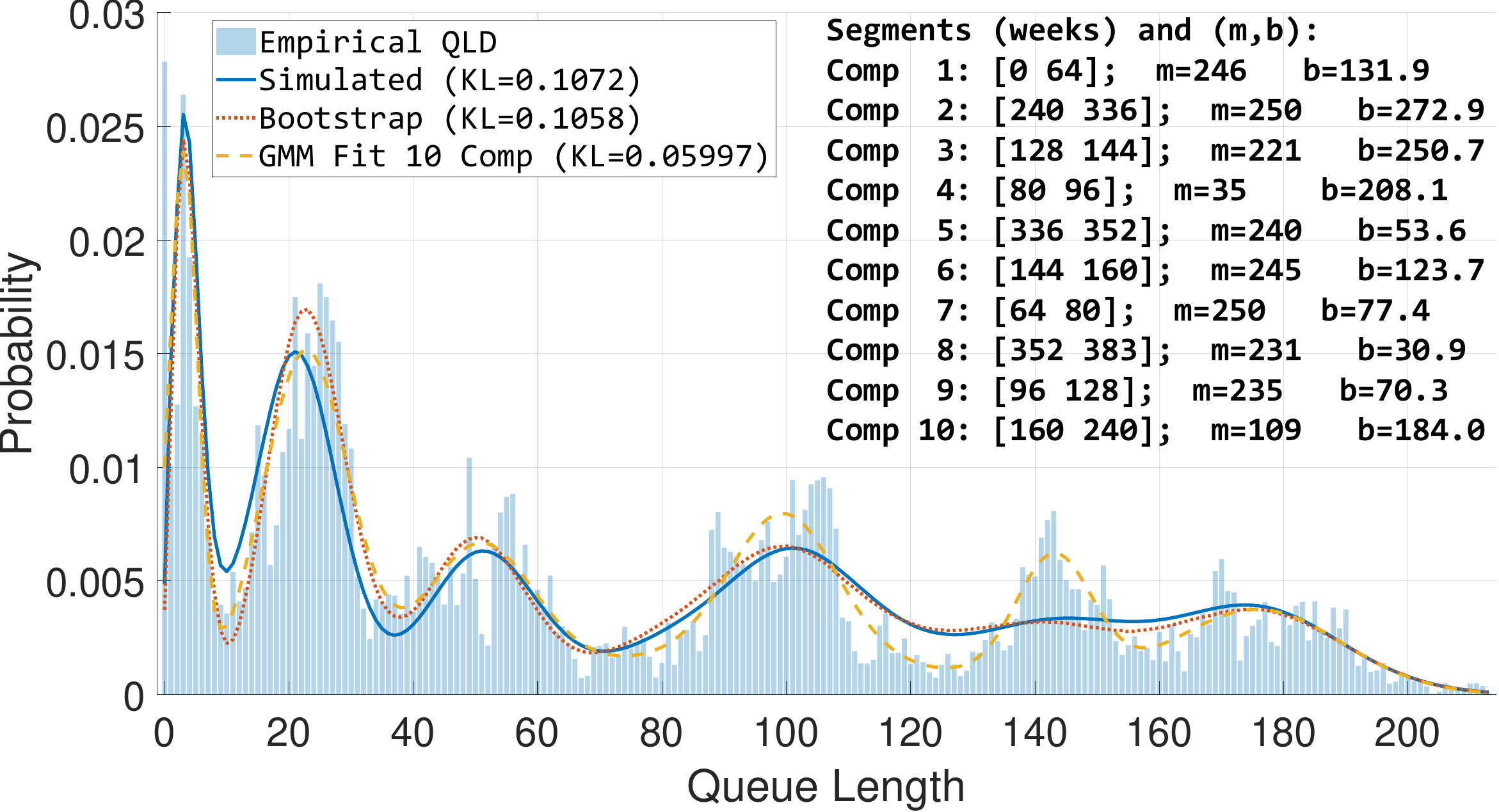}
    \caption{Final integrated model: empirical QLD compared with segmented, bootstrap, and 10-component GMM fits. In the segmented $G/G/m$–$b$ model, $b$ denotes the mean available defensive resource rather than the maximum capacity used in the global formulation.}
    \label{fig:final_qld}
    \vspace{-0.1in}
\end{figure}

\subsection{Temporal Characterization of Vulnerabilities of Software Releases}
\label{subsec:temporal}
The empirical fits above confirm that vulnerability service times follow a heavy-tailed distribution, with decay exponents ($u$) in the range \( 2 < u < 3 \). Consequently, vulnerabilities tend to remain active on the attack surface for extended periods, depending on the underlying defense and exploit dynamics. Our results are corroborated by prior measurements~\cite{lippmann2005evaluating} on specific networks, in which similar observations were made that the time a vulnerability remains exploitable follows a heavy-tailed law with a tail in \( D_s \) that decays slowly in a non-exponential fashion. Building on this empirical evidence, we formally show that when the vulnerability patching process exhibits such heavy-tailed behavior, the attack surface size \( N(t) \) develops \textbf{long-range dependence (LRD)}, even if the vulnerability arrival process itself is memoryless.

To formalize this observation, we model the vulnerability dynamics using an
$M/G/\infty$ queue. Recall that $V(t)$ denotes the vulnerability arrival process, and assume that it is a Poisson process with rate $\lambda$. Let $F(t)$ denote the cumulative distribution function of the service time $D_s$\footnote{For a vulnerability, $D_s = \min(D_d, D_e)$.}. The following result characterizes the temporal dependence structure of the resulting attack surface process $N(t)$.

\begin{thm}\label{thm:LRD}
Consider an $M/G/\infty$ queue in operation since time $t=-\infty$. Then, the queue-length process $N(t)$ is wide-sense stationary, and its covariance function is
\begin{equation}
\Cov(N(t),N(t+h)) = \lambda \int_h^\infty \bigl(1-F(s)\bigr)\,ds, \qquad h \ge 0.
\end{equation}
Furthermore, if the service-time tail is regularly varying, i.e.,
\begin{equation}
1-F(s) \sim L(s)s^{-\alpha}, \qquad 1<\alpha<2,
\end{equation}
where $L(\cdot)$ is a slowly varying function:
\[
\lim_{s\to\infty}\frac{L(cs)}{L(s)}=1, \qquad \forall c>0,
\]
then
\begin{equation}
\Cov(N(t),N(t+h)) \sim \frac{\lambda}{\alpha-1}L(h)h^{-(\alpha-1)}, \qquad h\to\infty,
\end{equation}
and therefore $N(t)$ exhibits long-range dependence.
\end{thm}

\begin{proof}
    See Appendix~\ref{app:LRD} for the complete proof. Here, we provide a sketch: our derivation begins by establishing the autocovariance function of the number of customers in an M/G/$\infty$ queue. We prove that $N(t)$ is wide-sense stationary and for heavy-tailed $F(s)$, the autocovariance function decays slower than $(\tau-t)^2$, leading to a long-range dependent attack surface.
\end{proof}

LRD implies that the impact of vulnerabilities persists over long time scales, and the system does not quickly forget past events. As a result, average-based or short-term observations may fail to accurately capture the true risk level of the system.
Even when the underlying system is the same, repeated observations of the process may show very different behavior over time, due to the persistent effects of past events.
This highlights the need for consistent and proactive mitigation strategies. For example, maintaining regular vulnerability discovery and patching cycles can help reduce long-tail effects on the defense side, thereby mitigating the impact of LRD in practice.

\section{A Near-Optimal RL Algorithm for Adaptive Defense}
\label{sec:RL}

Building on the model (Sec.~\ref{sec:model}) and the analytical insights into static allocation, temporal dependence, and AI-driven dynamics (Secs.~\ref{sec:static} and~\ref{sec:arvo}), we now build a systematic approach to the dynamic defense problem. To that end, we present a near-optimal RL algorithm for adaptively allocating constrained resources to dynamic defense, while accounting for switching costs. 
This learning-based approach is needed because the arrival and service processes governing the attack surface are unknown and may change over time.
The RL agent represents an adaptive defender who episodically reallocates limited patching or monitoring resources across a dynamic vulnerability queue. The transition dynamics of the attack surface process are \emph{unknown}. The policy switch for each episode corresponds to an operational reconfiguration (e.g., retuning patching pipelines or reassigning response teams), hence incurring measurable overhead modeled as a switching cost.

\subsection{Problem Setting and RL Framework}
\label{sec:learning}

Specifically, this section provides the solution to the constrained optimization problem stated in~(\ref{eq:rl_obj}), where the defender seeks the optimal dynamic defense policy under uncertainty.

\textbf{\textit{Setting:}} 
We apply the episodic Markov decision process (MDP) to model the dynamic defending problem. As assumed in standard episodic MDPs, we consider $H$ steps of interaction between the defender and the attack surface in each episode\footnote{Compared to the fixed allocation in traditional settings, this is a finer-grained setting where the defense action is taken in a more dynamic way.} $t=1,\ldots,T$. At each step $h=1,...,H$, based on the observed state $N^h(t)$ of the system, the defender can take a defense action $\mu_d^h(t)$ according to a policy $\pi_t: \gN \rightarrow \bm{\hat{\mu}}_d$, where $\gN = \{0, 1,\ldots,N\}$ is the system state space and $\bm{\hat{\mu}}_d = \{\mu_d\}$ is the defense action space. The defense must satisfy the resource-budget constraint $\mu_d^{h}(t) \leq b$ for all $h$ and $t$. After the defense action, the state evolves according to Eq.~(\ref{eq:evolution}).
The episodic formulation reflects practical operation cycles, such as periodic defense planning or monitoring windows, during which defense rates are adjusted based on observed backlog.

Given any state $N^h$ at step $h$, the defense action given by the current policy $\pi_t(\cdot)$ could be different from that given by the policy $\pi_{t-1}(\cdot)$ in last episode, in which case there will be a switching cost 
\begin{align}
\gS^h(t,N^h) \triangleq w \cdot \left | \pi_t(N^h) - \pi_{t-1}(N^h)\right | \nonumber
\end{align}
penalizing the change in defense across episodes.
The switching cost represents the operational overhead of changing defense actions between consecutive time steps.
In practice, this may correspond to reconfiguration effort, system adjustment, or resource reallocation required to implement a new defense policy.
The parameter $w$ controls the relative importance of the switching cost, and can be interpreted as a scaling factor that reflects how costly it is to change defense actions.
The magnitude of the switching cost is determined by the difference between consecutive defense actions, capturing how much the policy changes rather than simply whether a change occurs.

Therefore, the goal is to find a desirable algorithm $\pi$ that optimizes the expected cumulative penalty and defense cost over all steps and time-slots by executing the policies $\pi_{1:T}$, i.e., $\min_{\pi_{1:T}} \sum_{t=1}^{T} \left[ \gV^{\pi_t} + \gS^{\pi_t} \right]$. Here, the $\gV$-value function is defined to be
\begin{align}
\gV^{\pi_t} \triangleq \E \left[ \sum\nolimits_{h=1}^{H} \left[ C\left( N^h(t) \right) + \mu_d^{h}(t) \right] \right], \nonumber
\end{align}
where the expectation is taken with respect to the randomness of the state transition (\ref{eq:evolution}) and the race condition, and with a slight abuse of notation, the total switching cost is defined to be
\begin{align}
\gS^{\pi_t} \triangleq \E \left[ \sum\nolimits_{h=1}^{H} \sum\nolimits_{N^h \in \gN} \gS^{h}(t,N^h) \right]. \label{eq:defineswitching}
\end{align}
\textbf{\textit{Performance Metric:}} We use the standard regret as the metric to evaluate the performance of RL algorithm $\pi$, which is defined to be 
\begin{align}
\regret(T) \triangleq \sum\nolimits_{t=1}^T \left[\gV^{\pi_t} + \gS^{\pi_t} - \gV^* \right], \label{eq:defineregret}
\end{align}
i.e., the difference between the expected cumulative cost of the RL algorithm $\pi$ and the expected cumulative cost $\gV^*$ of the optimal policy $\pi^* = \argmin_{\{\pi:\pi(N^h) \leq b, \forall h, N^h\}} \gV^{\pi}$. Note that the optimal policy knows all problem parameters and does not change the policy. Thus, there is no switching cost and we drop the round index $t$. 
Intuitively, this regret measures the cumulative excess vulnerability exposure incurred by the learning defender relative to an omniscient optimal defense strategy.

\textbf{\textit{Novelties and Challenges:}} To our knowledge, this work is the first to analyze RL with switching costs that quantify the \emph{magnitude} of policy change rather than merely the frequency of switching. Specifically, the term $\gS^h(t,N^h)$ captures the absolute difference between consecutive defense actions, measuring how much the policy changes over time. In contrast, existing RL formulations penalize only whether $\pi_t$ differs from $\pi_{t-1}$, without accounting for the extent of change~\cite{bai2019provably,huang2022towards,shi2023nearswitch}. Moreover, we address the new challenge arising from the simultaneous presence of switching costs and resource-budget constraints, whose coexistence has not been studied in prior RL literature.

\subsection{Algorithm Design}
\label{subsec:algo_design}

Our algorithm maintains two Q-value estimates. One is updated continuously to learn from new data, while the other is updated less frequently to determine defense actions and limit switching costs, which is detailed in Algorithm~\ref{alg:pre1alg1}.
The algorithm outlines the core RL update under delayed policy switching. The algorithm maintains an optimistic Q-estimate, updated periodically according to a geometrically increasing triggering sequence to balance responsiveness and stability. From a high-level point of view, we take the defense action according to an optimistic belief-$\gQ$-value function. Specifically, at each step, our algorithm first updates an estimate-$\tilde{\gQ}$-value function, which represents the value of taking a certain action at a state (line 9). An action with larger estimate-$\tilde{\gQ}$-value function output is preferred. Intuitively, the belief-$\gQ$-value function should be updated more frequently when the sample size is small (i.e., the uncertainty is large), and it should be updated less and less frequently when the sample size becomes larger and larger. To achieve this, a delayed belief-$\gQ$-value function is updated when it has not been updated sufficiently long and triggers a switching threshold. Hence, to achieve effective defense with low switching costs, we need to carefully construct an effective belief-$\gQ$-value function to guide the defense action and construct an elegant triggering time sequence $\{t_n\}_{n\geq 1}$ for updating the belief-$\gQ$-value function (lines 11-12), as well as for changing defense actions.
Operationally, this design avoids frequent reconfiguration while still allowing rapid adaptation when uncertainty is high.

\begin{algorithm}[t]
\caption{Dynamic Defense Under Resource Constraints and Switching Costs}\label{alg:pre1alg1}
\begin{algorithmic}[1]
\STATE \textbf{Parameters:} $\eta = \frac{1}{2 H(H+1)}$ and $c>0$
\STATE \textbf{Initialization:} $\gQ$-value functions $\tilde{\gQ}^h(N,\mu_d) = H$ and $\gQ^h(N,\mu_d) = \tilde{\gQ}^h(N,\mu_d)$, state-action visitation counts $\gN^h(N,\mu_d) = 0$, where $N$ represents the number of vulnerabilities in the queue, and $\mu_d$ represents the defense action
\FOR{$t=1:T$}
\FOR{$h=1:H$}
\STATE Take defense action 

\hspace{0.5in} $\mu_d^h(t) = \argmin_{\{\mu_d\}} \gQ^h(N^h(t),\mu_d)$ 
\STATE Based on the arrivals of vulnerabilities and race condition, the queue state evolves to $N^{h+1}(t)$ according to Eq.~(\ref{eq:evolution})
\STATE Update the defense changing parameter 

\hspace{1.2in} $k = \gN^h(N^h(t),\mu_d^h(t))+1$
\STATE Update bonus $\mathcal{B}(k) = c\sqrt{H^3 /k}$ for defense exploration
\STATE Update the estimate-$\tilde{\gQ}$-value function according to Eq.~(\ref{eq:updateQvalueinalgorithm}).
\STATE Update the estimate-$\tilde{\gV}$-value function 

$\tilde{\gV}^h(N^h(t)) = \min\left\{ H, \max_{\{\mu_d\}} \tilde{\gQ}^h(N^h(t),\mu_d) \right\}$
\IF{$t \in \{t_n\}_{n\geq 1}$}
\STATE Update the belief-$\gQ$-value 

\hspace{1.0in} $\gQ^h(N^h(t),\cdot) = \tilde{\gQ}^h(N^h(t),\cdot)$
\ENDIF
\ENDFOR
\ENDFOR 
\end{algorithmic}
\end{algorithm}

Specifically, in Line 7 of Algorithm~\ref{alg:pre1alg1}, we first update the number of times the state $N^h(t)$ and defense action $\mu_d^h(t)$ are visited simultaneously, which will generate the bonus term in Line 8. This bonus term essentially captures the level of uncertainty after collecting $k$ samples, such that according to the concentration inequality (e.g., Hoeffding's inequality), the estimate-$\tilde{\gQ}$-value function is an optimistic estimate of the optimal true $\gQ$-value with high probability. To guarantee this, we update the estimate-$\tilde{\gQ}$-value function as follows,
\begin{align}
& \tilde{\gQ}^h(N^h(t),\mu_d^h(t)) = (1-\alpha(t))\tilde{\gQ}^h(N^h(t),\mu_d^h(t)) \nonumber \\
&\hspace{0.3in} + \alpha(t) \left[(\bar{C}+b -C(N^h(t))-\mu_d^h(t))/(\bar{C}+b) \right. \nonumber \\
&\hspace{1in} \left. + \tilde{\gV}^{h+1}(N^{h+1}(t))+\mathcal{B}(k)\right], \label{eq:updateQvalueinalgorithm}
\end{align}
where $\bar{C} = \sup \{C(\cdot)\}$. This estimate-$\tilde{\gQ}$-value is an weighted average between the old estimate-$\tilde{\gQ}$-value $\tilde{\gQ}^h(N^h(t),\mu_d^h(t))$ (for exploiting the knowledge learned from historical samples) and the newly learned knowledge from currently visited state-action pair
\begin{align}
&\left[ ( \bar{C}+b-(C(N^h(t))+\mu_d^h(t)) ) / (\bar{C}+b) \right] \nonumber \\
&\hspace{1.0in} + \tilde{\gV}^{h+1}(N^{h+1}(t)),
\end{align}
together with a bonus term $\mathcal{B}(k)$ (for encouraging exploration of potentially better defense strategies).

Finally, lines 11 and 12 determine whether or not to change the belief-$\gQ$-value function that will directly determine the defense action $\mu_d(t)$. Let $\tau(i) = \left\lceil(1+\epsilon)^i\right\rceil$ for $i=1,2,...$, and define the triggering time sequence as 
\begin{align}\label{eq:triggertime}
\left\{t_n\right\}_{n \geq 1}=\left[1, \tau\left(i_0\right)\right] \cup\left\{\tau\left(i_0+1\right), \tau\left(i_0+2\right), \ldots\right\},
\end{align}
where $\epsilon$ and $i_0 = \left\lceil\frac{\log \left(10 H^2\right)}{\log (1+\eta)}\right\rceil$ are hyper--parameters chosen by the algorithm. For all $t \in\{1,2, \ldots\}$, $\tau_{\text {last }}(t):=\max \left\{t_n: t_n \leq t\right\}$ and $\alpha(t)=\frac{H+1}{H+t}$. The triggering time sequence (\ref{eq:triggertime}) allows policy switch every time-slot at the beginning, and then the delay for policy switch keeps exponentially-increasing after a certain amount $\tau(i_0)$ of samples has been collected. For example, the policy switches as follows. Given state $N^h(t)$ at step $h$, we take some particular defense $\mu_d^h(t)$ for time $t$, and update both the estimate-$\tilde{\gQ}$-value and the belief-$\gQ$-value immediately. After the time-slot $\tau(i_0)$, we still update $\tilde{\gQ}$ immediately. However, we only update $\gQ$ when $t$ is in the triggering time sequence. This exponentially delayed update schedule enables high responsiveness early on and stability as uncertainty decreases, effectively balancing adaptation and switching cost.

\subsection{Theoretical Regret Bound: Near-Optimality}

We show that the proposed algorithm achieves near-optimal regret with high probability.
In particular, with high probability $1-p$, the theoretical regret of our algorithm is upper-bounded by $\tilde{O}(\sqrt{T})$, which is optimal. Recall that the regret is defined to compare our RL performance with the optimal policy, which is an oracle defender with full knowledge of system parameters and no switching penalty.

\begin{thm}
\label{thm:regret}
\textbf{(Regret Upper-Bound)} With high probability $1-p$, $p \in (0,1)$, the regret of Algorithm~\ref{alg:pre1alg1} is upper-bounded by $\tilde{O}\left(\sqrt{H^3 \bar{C}^4 b T}\right)$ for any horizon $T = \tilde{\Omega}\left(H^6 \bar{C}^2 b^2\right)$.
\end{thm}

\begin{proof}
See Appendix~\ref{app:pftheoremregret} for the complete proof. Here, we provide a sketch. The proof follows optimism-based analysis for episodic RL, with new developments to handle the delayed defense switching and resource budgets. Let $\tilde{\gQ}^h(t)$ denote the \emph{estimate-$\tilde{\gQ}$-value function}, which is continuously updated from new samples, and let $\gQ^h(t)$ denote the \emph{belief-Q-value function}, a stabilized version used for policy decisions and updated only at triggering times $\{t_n\}_{n\ge1}$. The proof involves the following key ideas.

\emph{(i) Optimism and Concentration:}
Each $\tilde{\gQ}^h(t)$ update uses a step size $\alpha(t)=\frac{H+1}{H+t}$ and an exploration bonus 
$\gB(k)=c\sqrt{H^3\log(1/p')/k}$. 
By standard concentration arguments (e.g., Azuma–Hoeffding inequality), with probability at least $1-p'$, the estimate satisfies
\[
0 \le \tilde{\gQ}^h(t) - \gQ^{*,h} 
\le \gB(k) = \tilde{O}\!\left(\sqrt{\tfrac{H^3}{k}}\right),
\]
uniformly over all $(h,N,\mu,t)$, where $\gQ^{*,h}$ is the optimal $\gQ$ value. This ensures \emph{optimism}: the learned Q-values upper bound the true optimal values within $\gB(k)$.

\emph{(ii) Regret Decomposition:}
Let $\delta^{h}(t) = \tilde V^{h}(t) - V^{\pi_t,h}$ denote the instantaneous regret at step $h$ in episode $t$. 
Since the policy $\pi_t$ is derived from the most recently updated $\tilde{\gQ}^h(k')$ at trigger $k' \!\le\! t$, we decompose
\[
\delta^{h}(t)
\le | (\gQ^{h}(k')-\gQ^{\pi_t,h}) | + | \tilde{\gQ}^{h}(k')-\gQ^{h}(k') | ,
\]
where the first term behaves as in standard optimistic RL, while the second term measures the deviation caused by delayed updates.

\emph{(iii) Controlling the Delay via Triggering Sequence:}
Between two triggers, $\tilde{\gQ}$ evolves according to small step sizes and bounded bonuses.
Under the geometrically increasing triggering schedule $t_n=\lceil(1+\epsilon)^n\rceil$, 
the cumulative deviation $\sum_t | \tilde{\gQ}^h(k')-\gQ(^hk') |$ grows at most by a constant factor $1+O(1/H)$ relative to non-delayed updates.
Hence, the delay introduces only a multiplicative $O(1/H)$ overhead.

\emph{(iv) Error Propagation over the Horizon:}
Summing the per-step inequalities and propagating value errors through the horizon yields
\[
R^h \;\le\; (1+O(1/H))R^{h+1} + \tilde{O}(\sqrt{H^3 T}) ,
\]
where $R^h=\sum_t\delta^{h}(t)$. 
Unrolling across $h=1,\dots,H$ gives
$\sum_h R^h = \tilde{O}(\sqrt{H^3 T})$. 
Including bounded per-step costs $\bar{C}$ and feasible budget $b$ scales the bound to 
$\tilde{O}(\sqrt{H^3 \bar{C}^4 b T})$.

\emph{(v) Switching Costs:}
The number of belief-$\gQ$-value updates, and hence policy changes, is logarithmic in time. Thus, the cumulative switching cost contributes at most $\tilde{O}(\log T)$ to regret, absorbed by the main term.

Combining the above, the total regret then follows.

\end{proof}

\vspace{-0.3cm}

To the best of our knowledge, this is the first regret bound established for dynamic defense 
under the coexistence of resource-budget constraints and switching costs for amount of changes.

\section{Numerical Evaluation}
\label{sec:sim}

We now evaluate the proposed framework through a series of numerical experiments that combine model-based simulations and data-driven evaluations. 
The goal of this section is twofold: first, to verify that the analytical insights derived from the queueing model are reflected in controlled settings, and second, to demonstrate that the proposed RL-based defense policy remains effective under realistic vulnerability dynamics.
In particular, we focus on how adaptive defense allocation influences the evolution of the attack surface, measured by the queue length $N(t)$, and the resulting exploit rates.
The queue length represents the number of unresolved vulnerabilities at a given time and therefore directly captures the system's exposure to attacks.
Reducing both its average level and its high-percentile tail is critical for limiting sustained cyber risk.
The analysis proceeds in three parts: (i) controlled model-based simulations to verify core dynamics, (ii) trace-driven evaluation using the ARVO dataset, and (iii) aggregate-budget simulations that examine RL reallocation under realistic resource constraints.

\subsection{Model-based RL evaluation}
\label{sec:synthetic}
We begin with controlled simulations based on the analytical model introduced in Section~\ref{sec:problem}. 
These simulations evaluate the RL defense policy in a simplified environment where all system parameters are known and the dynamics are fully specified. 
The objective here is not to model a realistic system, but to isolate the effect of adaptive control and understand how the RL policy behaves under different arrival patterns.
In particular, these experiments highlight two main effects: (i) the RL policy’s ability to reduce successful exploits compared to fixed allocations, and (ii) its ability to stabilize system behavior under nonstationary vulnerability arrivals.

\begin{figure}[t]
  \centering
  \includegraphics[width=0.65\linewidth]{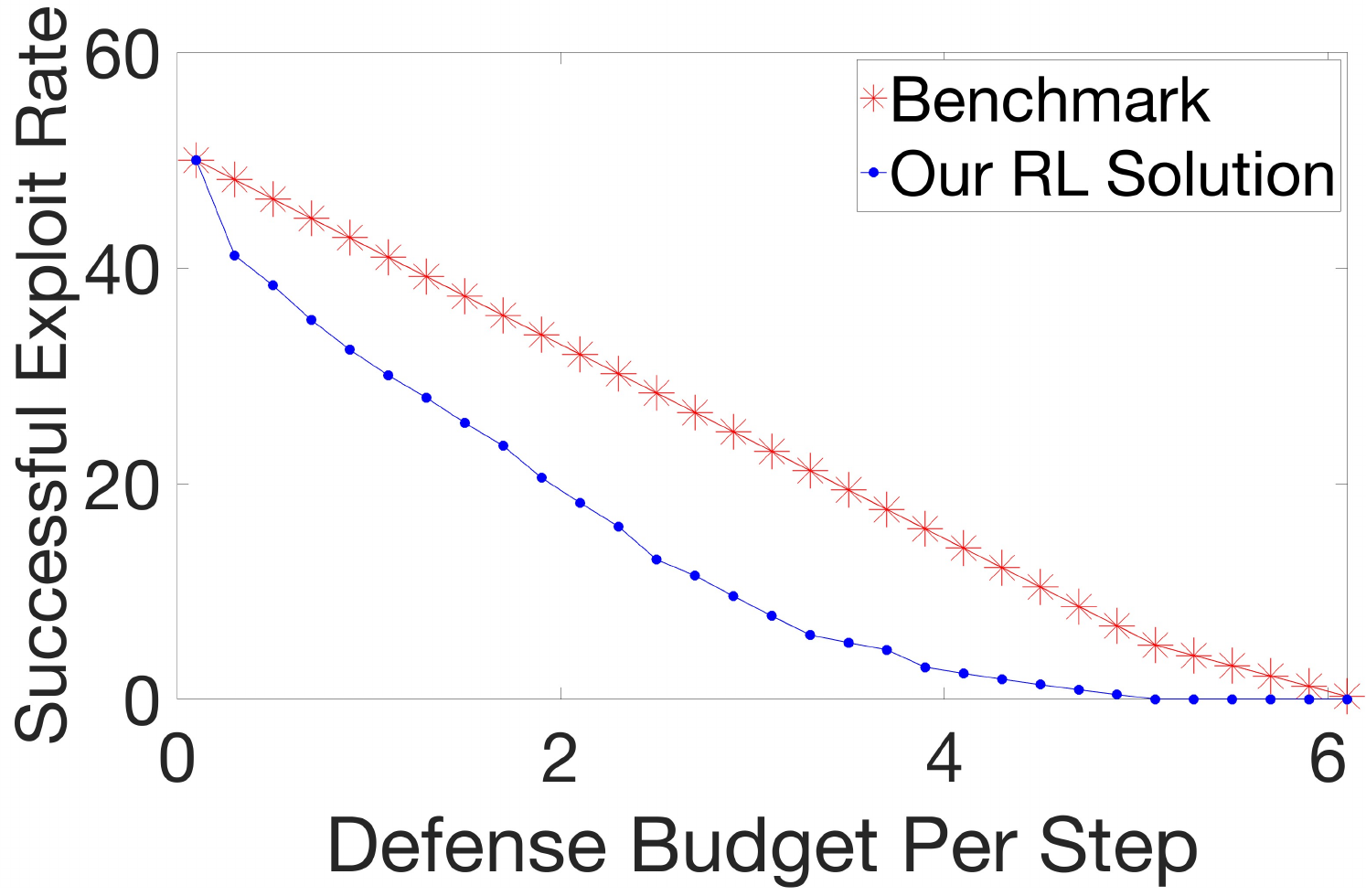}
  \caption{Successful exploit rate vs.\ per-step defense budget (patches per unit time) under stochastic arrivals.}
  \label{fig:attack1}
\end{figure}

\begin{figure}[t]
  \centering
  \includegraphics[width=0.65\linewidth]{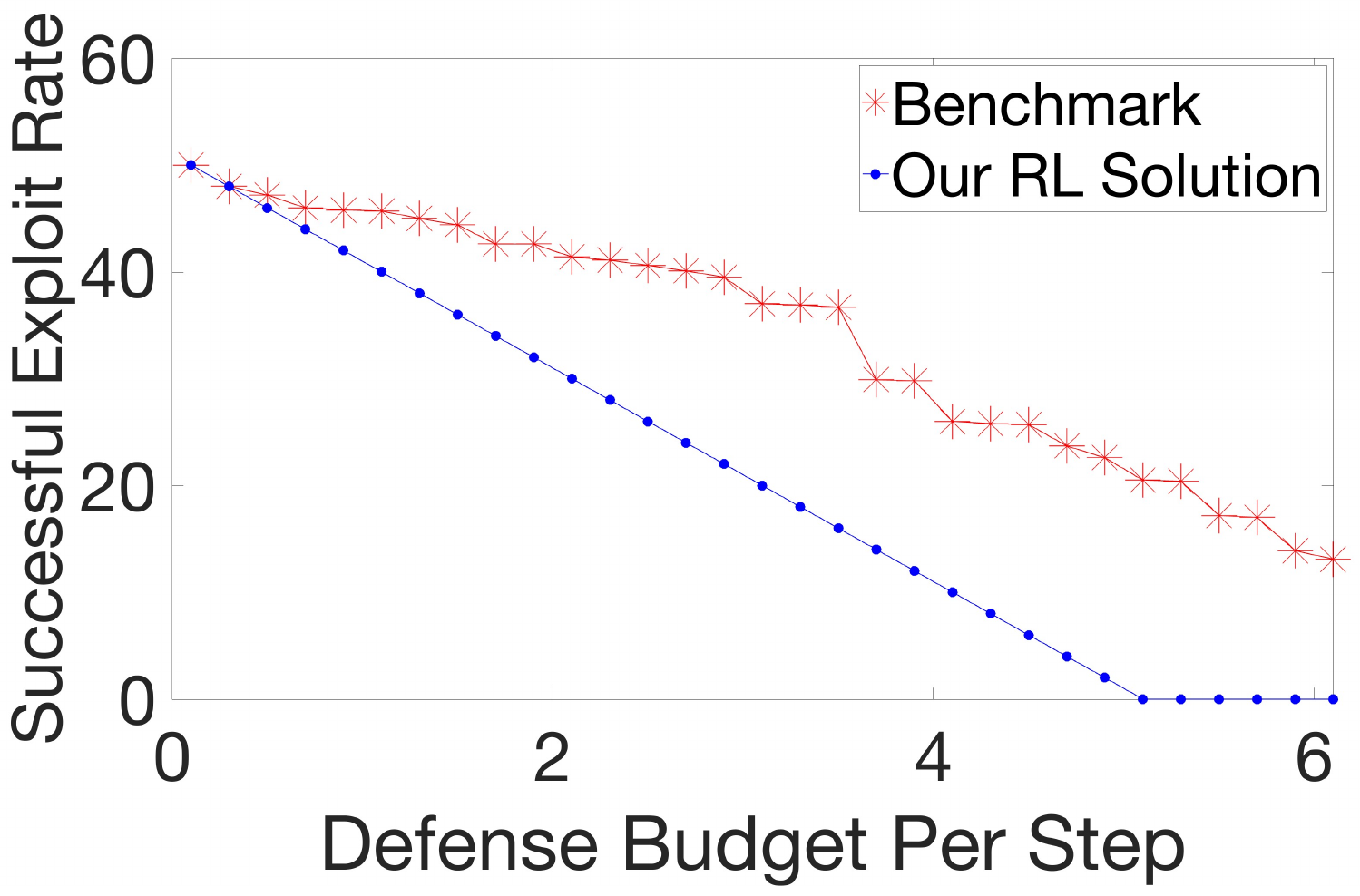}
  \caption{Successful exploit rate vs.\ per-step defense budget (patches per unit time) under adversarial arrivals.}
  \label{fig:attack2}
\end{figure}

\textbf{\textit{Setup:}} We consider $T=10^4$ rounds, each with $H=10$ steps, and use the model $\mu_\ell(t)=3N(t)$ for exploitation rate in these runs. Vulnerability arrivals are drawn under two regimes: (i) \emph{stochastic} arrivals with rate $\lambda=5$ at every step, and (ii) \emph{adversarial} arrivals that vary arbitrarily in $[0,10]$. We compare a fixed defending policy (constant per-step $\mu_d$) to the learned policy produced by our RL procedure, sweeping the per-step defense budget $b$ (measured in patches per unit time) and plotting the resulting successful exploit rates.

\textbf{\textit{Results:}} Figure~\ref{fig:attack1} and~\ref{fig:attack2} show successful exploit rate versus defense budget for the two arrival regimes. The learned policy reduces successful exploits substantially across budgets and attains up to a $55\%$ reduction for certain budget points in the stochastic regime. 
In the adversarial regime, the learned policy both reduces the mean exploit rate and smooths high-frequency fluctuations compared to the fixed allocation, reducing both the mean exploit rate and its variability compared to the fixed allocation. The reduction in variability improving predictability of performance, thereby enables an organization to better plan their budget to achieve a specific goal against the attacks.

\subsection{Trace-driven RL evaluation}
\label{sec:exp_arvo}

We now evaluate the RL defense policy using a trace-driven simulator built from the ARVO dataset, which provides event-level records of vulnerability discovery and patching in real-world open-source software projects.
At each time step, vulnerabilities arrive according to the observed ARVO trace, and the empirical defended counts define the baseline defense behavior. The RL agent observes the current queue length $N(t)$, which represents the number of active (unpatched) vulnerabilities, and selects a defense-rate action.
This action corresponds to how much defensive effort is allocated at that time step.
The defense rate $\mu_d(t)$ selected by the RL policy is interpreted as an expected number of patches per unit time.
In the simulator, this rate is mapped to an integer number of defended vulnerabilities by drawing from a Poisson distribution and truncating to a nonnegative integer.
This mapping captures the stochastic nature of patching processes while preserving the meaning of $\mu_d(t)$ as a controllable resource allocation variable.
The per-step defense budget $b$ represents the maximum number of patches that can be applied within a single time step, and therefore directly limits $\mu_d(t)$.
In contrast to static allocation, the RL policy dynamically adjusts $\mu_d(t)$ over time depending on the current system state, allowing it to prioritize defense when the attack surface grows large.
Importantly, the ARVO trace does not explicitly include attacker exploit events. Therefore, in this trace-driven setting, the queue evolution is governed by vulnerability arrivals and defense actions, and the queue length serves as the primary measure of system exposure.
A smaller queue length indicates fewer active vulnerabilities and thus reduced attack surface.
We compare the RL policy against the empirical baseline, which corresponds to the observed patching behavior in the ARVO dataset.

\begin{figure}[t]
    \centering
    \includegraphics[width=0.9\linewidth]{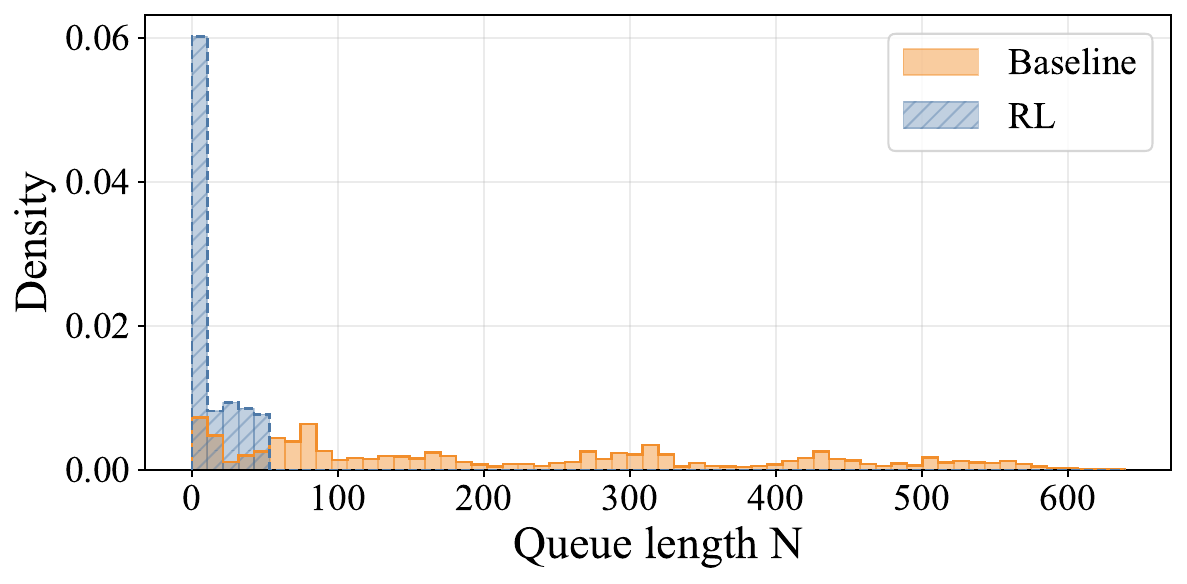}
    \caption{Trace-driven comparison of queue-length probability densities on the ARVO dataset (RL vs.\ baseline) for per-step budget $b=1.0$ patches per unit time.}
    \label{fig:arvo_hist}
\end{figure}

\begin{figure}[t]
    \centering
    \includegraphics[width=0.9\linewidth]{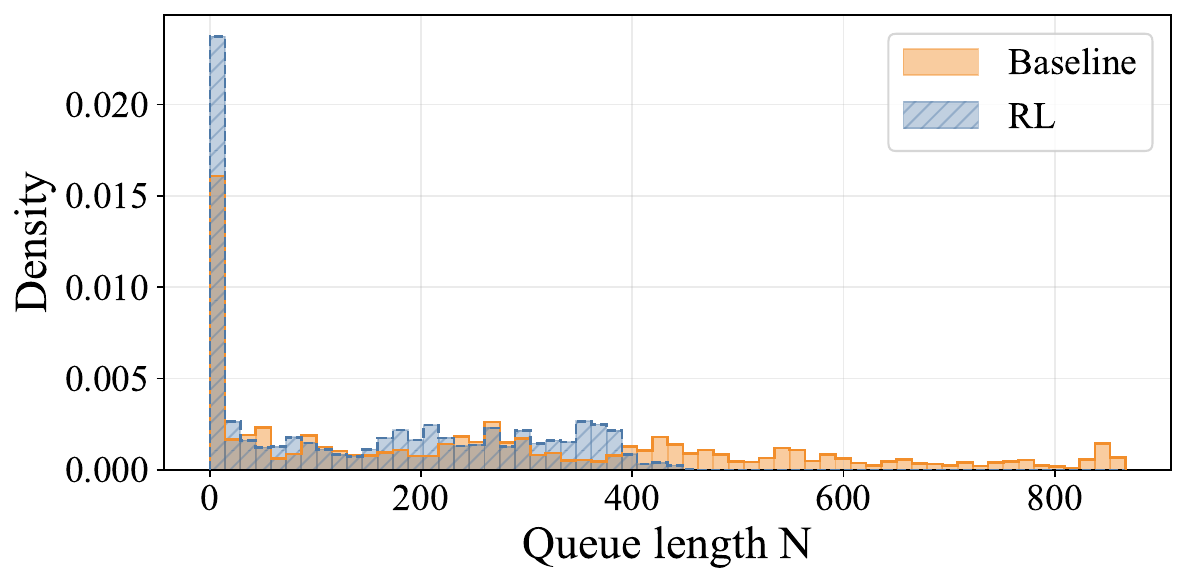}
    \caption{Queue-length histogram under the empirical aggregate baseline budget and the RL-reallocated policy. Statistics are reported in the text.}
        \label{fig:final_hist}
\end{figure}

\begin{table}[!t]
\centering
\caption{Trace-driven comparison statistics: queue-length mean \& variance (ARVO, RL vs.\ baseline).}
\label{tab:arvo_summary}
\begin{tabular}{lrr}
\toprule
Policy / Budget (patches per unit time) & Mean & Variance \\
\midrule
Baseline (data) & 219.9 & 30{,}930 \\
RL (b=0.5) & 59.4 & 1{,}602 \\
RL (b=1.0) & 13.0 & 219 \\
RL (b=1.5) & 4.7 & 56 \\
RL (b=2.0) & 1.2 & 6 \\
RL (b=2.5) & 0.1 & 0.4 \\
RL (b=3.0) & 0.1 & 0.3 \\
\bottomrule
\end{tabular}
\vspace{-0.5cm}
\end{table}

\textbf{\textit{Setup:}} We preprocess the ARVO data by binning records into 6-minute intervals, where each bin corresponds to one time step.
Ten consecutive bins form one episode (one hour), resulting in $T=64{,}395$ episodes.
For each bin, the number of reported vulnerabilities is used as the arrival count, and the empirical defended counts define the baseline defense events.
The RL agent observes only the queue length and selects actions from a discrete action set $\{0,0.5,1.0,1.5,2.0,2.5,3.0\}$.
Because the agent is table-based, we cap the indexed queue size at $N_{\max}=300$; larger values are recorded for evaluation but mapped to this index during decision making.
We evaluate multiple defense budgets $b\in\{0.5,1.0,1.5,2.0,2.5,3.0\}$, where $b$ represents the maximum patching rate per time step.
The RL algorithm includes an exploration bonus $\mathcal{B}(k)=c_{\mathrm{bonus}}\sqrt{H^3/k}$, and we use $c_{\mathrm{bonus}}=0.1$.
To ensure fair comparison, we tune the defense cost weight such that the total number of patches applied by the RL policy matches the empirical baseline on average.
All other algorithmic settings follow Section~\ref{subsec:algo_design}.

\textbf{\textit{Results:}} Figure~\ref{fig:arvo_hist} shows the queue-length density when per-step budget is $b=1.0$ patches per unit time.
For each budget $b$, we average results over 5 random seeds and report the queue-length mean and variance for the RL policy and the baseline in Table~\ref{tab:arvo_summary}.
The RL policy consistently reduces the queue length compared to the baseline, and the gains become larger as the per-step budget increases.
This result demonstrates the effectiveness of the RL policy in directly reducing the attack surface compared to the baseline under the same per-step budget.

\subsection{Aggregate-budget RL reallocation}
\label{sec:exp_integrated}

We next evaluate how much improvement can be obtained purely from better allocation of existing resources.
Specifically, we compare the baseline and RL policies when both operate under the same total defense budget over the entire time horizon.
The budget is obtained directly from the data-driven model in Section~\ref{sec:arvo}. 
Using the segmented analysis in Figure~\ref{fig:final_qld}, we estimate the baseline defense rate in each regime and sum these values to obtain the aggregate defense budget.
This represents the total amount of defense effort observed in the system.
The key difference is how this budget is used: the baseline applies it according to the observed (fixed) schedule, while the RL policy reallocates the same total budget across time steps depending on the system state. 
This setting isolates the benefit of adaptive allocation from the effect of increased resources.

\textbf{\textit{Setup:}} Following Figure~\ref{fig:final_qld}, we segment the ARVO trace into ten regimes and estimate the baseline per-segment defense rates.
In this subsection, we keep the ARVO arrivals unchanged and generate the baseline defended counts by drawing Poisson samples with the estimated per-segment rates.
We then sum the baseline's defense rates over the full period to obtain its aggregate defense budget.
The RL operates as in Section~\ref{sec:exp_arvo} but is constrained so that its total defense effort over the full period does not exceed this aggregate budget. This allows us to evaluate RL reallocation across steps under an equal total defense resource.

\textbf{\textit{Results (representative run):}} Figure~\ref{fig:final_hist} plots the queue-length densities for the baseline (with the estimated defense rates) and for RL (with aggregate-budget reallocation). 
RL substantially reduces large queue lengths compared to the baseline.
The summary statistics are $\mathrm{mean}(N)=146.63$, $N_{95}=379$, and $N_{99}=412.2$ for RL. $\mathrm{mean}(N)=267.52$, $N_{95}=772$, and $N_{99}=852$ for baseline.
Thus, RL reallocation substantially reduces the mean queue length and shrinks the high-percentile tails of the distribution.
In contrast to Figure~\ref{fig:arvo_hist}, where the per-step budget is fixed, this setting isolates the effect of resource allocation.
Both the baseline and RL operate under the same total defense budget, and the improvement is achieved purely by adapting how the budget is distributed over time.
This highlights that the performance gain is not due to increased resources, but due to more efficient allocation of the same resources.

\section{Discussion and Future Directions}
\label{sec:discussion}

We introduced a spatio-temporal queueing abstraction for the attack surface that models incoming vulnerabilities as arrivals and departures as either successful exploits or successful patches, and used this framework to derive several analytic insights. In particular, we highlight (i) a highly non-linear relationship between defense-resource shortfall and the rate of successful exploits, (ii) the emergence of long-range temporal dependence in the attack surface process when vulnerability lifetimes are heavy-tailed, and (iii) the fact that an aggregate AI-amplification of arrival and exploit rates can increase breach rates even when the attack surface distribution remains qualitatively similar.

While our analysis primarily focuses on a single-component system for clarity, the framework naturally extends to multi-component and multi-organizational settings. An organization’s total attack surface can be viewed as a collection of correlated queues, each representing a subsystem such as authentication, storage, or cloud services. At a larger scale, an ecosystem of organizations can be modeled as a network of statistically dependent queueing systems, capturing interdependencies arising from shared software libraries, third-party integrations, or supply-chain relationships. For tractability, we restrict our formal analysis to the single-queue case, which already exhibits the essential dynamics of heavy-tailed persistence, feedback coupling, and resource constraints that characterize real-world attack surface evolution.

Building on these foundations, natural directions for future work include extending the queueing abstraction to \textbf{multiple, dependent queues} that reflect component structure; modeling the ecosystem of multiple organizations (and their interactions) as an interconnected queueing system; exploring collaborative defense formulations (e.g., multi-agent approaches) under resource constraints; and expanding empirical data collection to strengthen and validate the modeling assumptions.

\section{Conclusion}

We developed a queueing-theoretic framework to model the cyber attack surface, accounting for both its temporal dynamics and spatial structure. Unlike models focused on isolated vulnerabilities, our system-level formulation shows how heavy-tailed patching times induce long-range dependence and persistent cyber risk. We also quantified how AI-driven capabilities shift the spatio-temporal dynamics of both attack and defense.

Using this core model, we formulated adaptive defense as a constrained sequential control problem and implemented an RL-based solution for dynamic resource allocation. This approach links data-driven system identification with control, allowing for decision-making under strict resource constraints.

We evaluate the proposed framework using both model-based simulations and trace-driven experiments. The results show clear improvements in the attack surface size under the same resource constraints. 
In the model-based experiments, the RL policy achieves up to 55\% reduction in exploit rates compared to fixed allocation. 
In the trace-driven setting, the average attack surface size is reduced by more than 90\%. 
Likewise, with fixed aggregate-budget, where both policies use the same total defense budget, this reduction is 45\%. Also, across all experiments, the tail (95-percentile) surface size is reduced by more than 50\%.
These results demonstrate that adaptive allocation of defensive resources facilitated by the novel queue model to the attack surface can significantly improve system performance without increasing the total budget.

\bibliographystyle{IEEEtran}
\bibliography{references}

\appendices
\section{Proof of Theorem~\ref{thm:LRD}}
\label{app:LRD}

Let us define the indicator variable:
\[ I_i(t) \stackrel{\triangle}{=} 
\left\{ \begin{array}{ll}
1, & \ {\text{a vulnerability arrives in $((i-1)\delta,i\delta)$}} \\ 
{} &\ \ \ {\text{ and still in the system at time $t$}} \\
0, & \ {\text{otherwise}}
\end{array} \right. .  \] 
for $\delta$ being the amount of temporal increment. Therefore, \vspace{-0.1in}
\begin{equation}
\label{noft}
N(t)=\sum_{i=-\infty}^{t/\delta} I_i(t) . \vspace{-0.1in}
\end{equation}
For any $t,\tau$ such that $t \leqslant \tau$, \vspace{-0.1in}
\begin{align}
\nonumber
&\EE{N(t)N(\tau)} = \EE{\sum_{i=-\infty}^{t/\delta}
\sum_{j=-\infty}^{\tau/\delta} I_i(t) I_j(\tau)} \\
\nonumber
&= \sum_{\{ i\leqslant t/\delta ,j \leqslant \tau/\delta \ |\ i\neq j \}}
\EE{I_i(t)I_j(\tau)} + \sum_{i=-\infty}^{t/\delta} \EE{I_i(t)I_i(\tau)} \\
\label{indarrivals1}
&= \sum_{\{ i\leqslant t/\delta ,j \leqslant \tau/\delta \ |\ i\neq j \}}
\EE{I_i(t)}\EE{I_j(\tau)} + \sum_{i=-\infty}^{t/\delta} \EE{I_i(t)I_i(\tau)} , \vspace{-0.2in}
\end{align}
where Eq. (\ref{indarrivals1}) follows since $I_i(t)$ and $I_j(\tau)$ are
independent for $i\neq j$. We can also write \vspace{-0.1in}
\begin{align}
\nonumber
\EE{N(t)}&\EE{N(\tau)} = \EE{\sum_{i=-\infty}^{t/\delta} I_i(t)}
\EE{\sum_{j=-\infty}^{\tau/\delta} I_j(\tau)} \\
\nonumber
&= \sum_{\{ i\leqslant t/\delta ,j \leqslant \tau/\delta \ |\ i\neq j \}} \EE{I_i(t)}\EE{I_j(\tau)} \\ 
\label{indarrivals2}
&\hspace{1in} + \sum_{i=-\infty}^{t/\delta} \EE{I_i(t)}\EE{I_i(\tau)} . \vspace{-0.1in}
\end{align}
Combining Eq.~(\ref{indarrivals1})~and~(\ref{indarrivals2}) we get, 
\begin{align}
\nonumber
&\cov{N(t),N(\tau)} = \EE{N(t)N(\tau)} - \EE{N(t)}\EE{N(\tau)} \\
\label{indarrivalmain1}
&\hspace{0.7in}= \sum_{i=-\infty}^{t/\delta} \left\{ \EE{I_i(t)I_i(\tau)}
- \EE{I_i(t)}\EE{I_i(\tau)} \right\} \\
\label{indarrivalmain2} &\hspace{0.7in} =
\sum_{i=-\infty}^{t/\delta} \left\{ \lambda \delta [1-F(\tau-i\delta)]
- o(\delta) \right\} , \vspace{-0.1in}
\end{align}
where the first term in (\ref{indarrivalmain2}) is due to $\EE{I_i(t)I_i(\tau)}$ and the o$(\delta)$ is due to $\EE{I_i(t)}\EE{I_i(\tau)}$ and that the probability of two vulnerabilities in the same instant is o$(\delta)$.
As $\delta \rightarrow 0^+$, 

\begin{align}
\cov{N(t),N(\tau)} & = \int_{-\infty}^{t} \lambda \left[ 1-F(\tau-s)\right] ds \notag \\
& = \lambda  \int_{| \tau - t |}^{\infty} \left[ 1-F(s)\right] ds.
\label{indarrivalmain3}
\end{align}

Also,
\[
\EE{N(t)}=\lambda\int_0^\infty [1-F(s)]\,ds,
\]
which is independent of $t$. Hence, the mean is constant. Since the covariance
in Eq. (\ref{indarrivalmain3}) depends only on the lag $|\tau-t|$, the process
$\{N(t)\}$ is wide-sense stationary.

Furthermore, we consider heavy-tailed service-time distributions with regularly
varying tail, i.e.,
\[
1-F(s)\sim L(s)s^{-\alpha}, \qquad 1<\alpha<2,
\]
where $L(\cdot)$ is slowly varying. Under this assumption,
\begin{align}
    \Cov(N(t),N(t+h))
    & =\lambda\int_h^\infty [1-F(s)]\,ds \notag \\
    & \sim \frac{\lambda}{\alpha-1}L(h)h^{-(\alpha-1)},
    \qquad h\to\infty.
\end{align}
Since $0<\alpha-1<1$, the autocovariance is non-integrable, and therefore
$\{N(t)\}$ exhibits long-range dependence.

\section{Proof of Theorem~\ref{thm:regret}}\label{app:pftheoremregret}

\begin{proof}

For readability, we remove some indices when the context is clear.
Recall that $\tilde{\mathcal Q}^{h}(t)$ denotes the \emph{estimate-$\gQ$-value function},
continuously updated from new samples at step $h$ and time $t$,
while $\mathcal Q^{h}(t)$ denotes the \emph{belief-$\gQ$-value function},
a stabilized version used to determine the defense action and updated
only at triggering times $\{t_n\}_{n\ge1}$.

\vspace{4pt}
\noindent\textbf{Step 1: Instantaneous regret decomposition:}
Let $\tilde{\delta}^{h}(t)$ be the instantaneous regret at step $h$ and time $t$.
Following the standard optimistic decomposition,
\begin{align}
\tilde{\delta}^{h}(t)
& = 
\Big(
\max\big\{\tilde{\mathcal Q}^{h}(t_\uparrow),\,
\tilde{\mathcal Q}^{h}(t)\big\}
- \mathcal Q^{*,h}
\Big)
\!\big(N^{h}(t),\mu_d^{h}(t)\big) \nonumber \\
& \leq 
|
\tilde{\mathcal Q}^{h}(t_\uparrow) - \mathcal Q^{*,h}
|
+
|
\tilde{\mathcal Q}^{h}(t) - \tilde{\mathcal Q}^{h}(t_\uparrow)
| ,
\label{eq:regret-decomp}
\end{align}
where $t_\uparrow=\tau_{\mathrm{last}}(t)+1$
denotes the most recent triggering time before $t$.
The first term corresponds to the standard value-estimation error,
and the second term represents the additional deviation introduced
by the delayed belief update.

\vspace{4pt}
\noindent\textbf{Step 2: Recursive relation for $\tilde{\mathcal Q}^{h}(t)$:}
The update of the estimate-$\tilde{gQ}$-value at visit $t$ can be expanded as
\begin{align}
& \tilde{\mathcal Q}^{h}_{t}(N,\mu_d)
- \mathcal Q^{\pi,h}(N,\mu_d) \nonumber \\
& =
\alpha(t_i)\!
\left(
H - \mathcal Q^{\pi,h}(N,\mu_d)
\right) + 
\sum_{i=1}^{k}
\alpha(t_i)
\Big[
C^{h}(N,\mu_d) \nonumber \\
&\qquad + \tilde V^{\,h+1}_{t_i}\!\big(N^{h+1}_{d}(t_i)\big)
- V^{\pi,h+1}\!\big(N^{h+1}(t_i)\big) \nonumber \\
&\qquad
+\,
\big(\widehat P^{h}_{t_i}-P^{h}\big)V^{\pi,h+1}(N,\mu_d)
+ \gB(t_i)
\Big],
\label{eq:q-update-expansion}
\end{align}
where $\alpha(t_i)$ is the step size,
$B(t_i)$ is the exploration bonus,
and $\widehat P^{h}_{t_i}$ is the empirical transition model.
This expression separates the stochastic update noise,
transition deviation, and optimism term.

\vspace{4pt}
\noindent\textbf{Step 3: Bounding the delayed perturbation:}
For any $t>t_\uparrow$, the cumulative drift between
two consecutive triggers satisfies
\begin{align}
& |
\tilde{\mathcal Q}^{h}(t)
- \tilde{\mathcal Q}^{h}(t_\uparrow)
|\!
\big(N^{h}(t),\mu_d^{h}(t)\big) \nonumber \\
& \le
\phi_k
+
\sum_{i=\tau_{\mathrm{last}}(t)+1}^{t}
\alpha(t)\,
\tilde\zeta^{\,h+1}_{t_i}
+ 
\bar\zeta^{\,h}_{t},
\label{eq:delay-perturb}
\end{align}
where $\phi_k=O\!\big(\sqrt{H^3/k}\big)$ and
$\bar\zeta^{\,h}_{t}=O\!\big(\sqrt{H^3/t}\big)$ hold
uniformly with high probability.
Both terms can be absorbed into a constant multiple of $\phi_k$,
since $\phi_{t_\uparrow}\le(1+O(1/H))\phi_k$ under the geometric
triggering rule.

\vspace{4pt}
\noindent\textbf{Step 4: Coefficient aggregation:}
When summing Eq.~(\ref{eq:delay-perturb}) over all time steps,
each successor-state error $\tilde\zeta^{\,h+1}_{t_\uparrow}$
is weighted by the accumulated step-size coefficients.
Using the triggering sequence
$t_n=\lceil(1+\varepsilon)^n\rceil$ with
$\varepsilon=\tfrac{1}{2H(H+1)}$
and initial index
$r_0=\big\lceil\tfrac{\log(10H^2)}{\log(1+\varepsilon)}\big\rceil$,
we obtain
\begin{align}
& \sum_{t}
\Big(
\mathbf{1}\{\text{no trigger at }t\}\,\alpha(\tau_{\mathrm{last}}(t))
+
\mathbf{1}\{\text{trigger at }t\}\,\alpha(t)
\Big) \nonumber \\
& \le 1+\tfrac{3}{H},
\end{align}
showing that the delay inflates the propagation factor
by at most $(1+3/H)$.

\vspace{4pt}
\noindent\textbf{Step 5: Regret recursion and final bound:}
Let $R^h=\sum_t\tilde{\delta}^{h}(t)$.
Combining the above results yields
\[
R^h \le (1+O(1/H))R^{h+1} + \tilde{O}(\sqrt{H^3 T}).
\]
Unrolling the recursion over $h=1,\dots,H$
gives $\sum_{h=1}^{H}R^h=\tilde{O}(\sqrt{H^3 T})$.
Accounting for the bounded per-step cost $\bar{C}$
and defense-cap budget $b$ scales the bound to
$\tilde{O}(\sqrt{H^3 \bar{C}^4 b T})$.
The number of belief-$\gQ$-value updates, and hence policy changes, is logarithmic in time. Thus, the cumulative switching cost contributes at most $\tilde{O}(\log T)$ to regret, absorbed by the main term.

\end{proof}

\end{document}

%% file: mathcommand.tex
\usepackage{amsmath,amsfonts,bm}

\def\eqref#1{equation~\ref{#1}}

\def\1{\bm{1}}

\DeclareMathAlphabet{\mathsfit}{\encodingdefault}{\sfdefault}{m}{sl}
\SetMathAlphabet{\mathsfit}{bold}{\encodingdefault}{\sfdefault}{bx}{n}

\def\gB{{\mathcal{B}}}

\def\gN{{\mathcal{N}}}

\def\gQ{{\mathcal{Q}}}

\def\gS{{\mathcal{S}}}

\def\gV{{\mathcal{V}}}

\newcommand{\E}{\mathbb{E}}

\newcommand{\Cov}{\mathrm{Cov}}

\DeclareMathOperator*{\argmin}{arg\,min}